\RequirePackage[]{silence}
\documentclass[11pt]{article}
\usepackage{fullpage}

\usepackage{amsmath, amsthm, amssymb}
\usepackage{graphicx}
\usepackage{epstopdf}
\usepackage{algorithm}
\usepackage[noend]{algpseudocode}
\usepackage{url}

\usepackage[pagebackref=false]{hyperref}
\hypersetup{
    unicode=false,          
    colorlinks=true,        
    linkcolor=red,          
    citecolor=green,        
    filecolor=magenta,      
    urlcolor=cyan           
}
\usepackage{cleveref}

\newtheorem{theorem}{Theorem}[section]
\newtheorem{lemma}[theorem]{Lemma}

\newtheorem{corollary}[theorem]{Corollary}

\newcommand{\fprog}{F_{prog}}
\newcommand{\fack}{F_{ack}}

\newcommand{\cmpdell}[1]{\textcolor{blue}{{}}}

\newif\ifold
\oldfalse

\newif\ifnew
\newtrue

\newif\iflong
\longfalse

\newcommand{\bcastq}[0]{bcastq}
\newcommand{\cR}[0]{{\cal R}}
\newcommand{\cM}[0]{{\cal M}}
\newcommand{\cC}[0]{{\cal C}}

\newcommand{\mm}[0]{m}

\newcommand{\BMMB}[0]{{\mbox{BMMB}}}

\newcommand{\FullOrShort}{full}

\ifthenelse{\equal{\FullOrShort}{full}}{
	
	  \newcommand{\fullOnly}[1]{#1}
	  \newcommand{\shortOnly}[1]{}

  }{
		\usepackage{times}

	  \newcommand{\fullOnly}[1]{}
	  \newcommand{\shortOnly}[1]{#1}
}

\begin{document}

\date{}
\title{Multi-Message Broadcast with\\ Abstract MAC Layers and Unreliable Links%
\thanks{Supported in a part by AFOSR FA9550-13-1-0042 and NSF grants Nos. CCF-0939370, CCF-1217506, and CCF-AF-0937274.}
}

\author{
  Mohsen Ghaffari\\
	\small MIT\\
  \small ghaffari@csail.mit.edu
  \and
	Erez Kantor\\
	\small MIT\\
  \small erezk@csail.mit.edu
	\and
	Nancy Lynch\\
	\small MIT\\
  \small lynch@csail.mit.edu
	\and
	Calvin Newport\\
	\small Georgetown University\\
  \small cnewport@cs.georgetown.edu
}

\maketitle

\maketitle

\begin{abstract}
We study the multi-message broadcast problem using abstract MAC layer models
of wireless networks. These models capture the key guarantees of existing MAC layers
while abstracting away low-level details such as signal propagation and contention.
We begin by studying upper and lower bounds for this problem in a {\em standard abstract
MAC layer model}---identifying an interesting dependence between the structure
of unreliable links and achievable time complexity. In more detail,
given a restriction that devices connected
directly by an unreliable link are not too far from each other in the reliable link topology, we can
(almost) match the efficiency of the reliable case.
For the related restriction, however, that two devices connected by an unreliable link are not
too far  from each other in geographic distance,
we prove a new lower bound that shows that this efficiency is impossible.
We then investigate how much extra power must be added to the model to enable
 a new order of magnitude of efficiency. In more detail, we consider an {\em enhanced abstract MAC layer
model} and present a new multi-message broadcast algorithm that (under certain natural assumptions)
solves the problem in this model
faster than any known solutions in an abstract MAC layer setting.
\end{abstract}

\shortOnly{


}
\section{Introduction}

Most existing work on distributed algorithms for wireless networks assumes low-level synchronous models
that require algorithms to deal directly with link-layer issues such as signal fading (e.g.,~\cite{Kesselheim:2010,jurdzinski:2013:random,daum:2013})
and contention (e.g.,~\cite{baryehuda:1987,gasieniec:2001,kowalski:2005,czumaj:2006}).
%
%
These low-level models are appropriate for answering {basic science} questions about the capabilities and limits
of wireless communication. We argue, however, that they are often {\em not} appropriate for designing and analyzing algorithms meant
for {real world deployment}, because: (1) they fail to capture the unpredictable reality of real radio signal propagation (which tends
not to follow simple collision or fading rules~\cite{newport:2007});
(2) they do not address issues like network co-existence (it is rarely
true that your algorithm is alone in using the wireless channel);
 and (3) they ignore the presence of general purpose MAC layers which are standard and hard to bypass in existing devices.

In~\cite{kuhn:2009,kuhn:2011abstract}, we introduced the {\em abstract MAC layer} approach as an alternative
to low-level models for studying wireless algorithms.
This approach moves the algorithm designer up the network stack by modeling the basic guarantees of most existing wireless MAC layers.
In doing so, it abstracts away
low level issues such as signal fading and contention,
instead capturing the impact of this behavior on higher layers with model uncertainty. 
Because abstract MAC layers are defined to maintain the basic guarantees of most standard wireless MAC layers,
algorithms developed in such models can be deployed on existing
devices while maintaining their formally proved properties.

In this paper, we study the basic communication primitive of {\em multi-message broadcast} (a subset of devices start with one or more messages
they need to disseminate to the whole network) in abstract MAC layer models that include unreliable links.
We produce new upper and lower bounds, and explore new model variants. Our results, summarized below and in Figure~\ref{fig:results},
provide new theoretical insight into the relationship between unreliability and efficiency,
 and identify practical algorithmic strategies.

\bigskip

\noindent {\bf Abstract MAC Layer Models.} Abstract MAC layer models provide devices with an acknowledged local broadcast
primitive that guarantees to deliver a message to a device's reliable neighbors (captured by a graph $G$)
and possibly some arbitrary
subset of additional unreliable neighbors (captured by a graph $G'$). At some point after the message is delivered, the sender receives
an acknowledgment.\footnote{The acknowledgment in this model
 describes the behavior of a standard MAC layer asking for the next message to broadcast
from the send queue; i.e., after a CSMA back-off protocol finishes with a given packet. 
It does not represent an acknowledgment explicitly sent from the receivers.} The performance of the model in a given execution is defined by two
constants: $\fack$, the maximum time for a given local broadcast to complete and be acknowledged,
and $\fprog$, the maximum time for a device to receive {\em some} message when at least one
nearby device is broadcasting.
We note that in both theory
and practice,
 $\fprog \ll \fack$.\footnote{From a theory perspective,
we note that
standard probabilistic strategies like {\em decay} (cycle through an exponentially decreasing series of broadcast probabilities),
when analyzed in graph-based, low-level wireless models,
offer $F_{prog}$ values that are polylogarthmic in the maximum possible contention,
while $F_{ack}$ values can be linear (or worse) in this same parameter (see~\cite{ghaffari:2012} for more discussion).
From a practical perspective, this gap is easily demonstrated.
Consider, for example, a star network topology centered on device $u$ where
all points in the star have a message to broadcast. If these nodes are using standard back-off style strategies,
$u$
will receive {\em some} message quickly. But regardless of how contention is managed,
there are some points in the star that will have to wait a long time (i.e., linear in the star size) for their turn.}

\bigskip

\noindent {\bf Results.}
In this paper, we consider the multi-message broadcast (MMB) problem. This problem distributes $k\geq 1$ messages to devices
at the beginning of an execution, where $k$ is not known in advance. It is considered solved once all messages are delivered to all nodes.
We begin by studying a natural MMB strategy called Basic Multi-Message Broadcast (BMMB)
in what we call the {\em standard abstract MAC layer} model: a basic model that captures the key guarantees
of existing MAC layers.
The BMMB algorithm implements an expected strategy for broadcast: on first receiving a message $m$, from the environment
or from another device, add it to your FIFO
queue of messages to broadcast; going forward, discard any additional copies of $m$ that you receive.
In previous work~\cite{kuhn:2011abstract}, we proved that BMMB solves the problem in $O(D\fprog + k\fack)$ time
 in the standard abstract MAC layer model
under the strong assumption that $G'=G$, i.e., there are no unreliable links, and $D$ is the diameter of $G$.
In the first part of this paper, we investigate the performance of strategy in the presence of unreliability.

\begin{figure*}[t]
\centering

\begin{tabular}{|c||c|c|c|}
\hline
Model/G'  Const. & $G' = G$ & $r$-Restricted & Grey Zone \\ \hline \hline
Standard & $O(D\fprog + k\fack)$~\cite{kuhn:2011abstract} & $O(D\fprog + r k \fack)$ & $\Theta((D+k)\fack)$  \\ \hline
Enhanced &  {\em same as grey zone} & {\em open} & $O((D + k\log{n} + \log^3{n})\fprog)$ \\ \hline
\end{tabular}
\label{fig:results}
\caption{\footnotesize  A summary of results categorized by model type and constraints assumed regarding $G'$.
With the exception of the $G'=G$ result for the standard model, the results
in this table are proved for the first time in this paper. We note that the grey zone result for the standard model summarizes
two separate results: an upper bound and matching lower bound.  These two results also hold for arbitrary $G'$.}
\vspace{-10pt}
\end{figure*}

We begin by considering the case where $G'$ is arbitrary; i.e., there are no constraints on the structure of unreliable links.
Under this pessimistic regime, we reanalyze BMMB,
proving a new guarantee of $O((D+k)\fack)$ time,
which is (potentially) much slower than what is possible when $G'=G$.
The presence of unreliable edges, it turns out,
breaks the careful induction at the core of the $G'=G$ result,
as they allow old messages to arrive unexpectedly from farther
away in the network at inopportune points in an execution.

Not satisfied with this slow-down,
we then consider the case of an $r$-{\em restricted} $G'$---a natural
constraint that only allows $G'$ edges between nodes within $r$ hops in $G$.
Under these constraints, we can now show that BMMB
solves the problem in $O(D\fprog + r\cdot k\cdot \fack)$ time,
which is close to the $G'=G$ case for small $r$.
This proof
discards the core strategy of the $G'=G$ case, which depends heavily on the lack of unreliable links,
and instead uses a new type of pipelining argument that carefully accounts for the possible message behavior
over $r$-restricted $G'$ links.

We conclude our investigation of BMMB
by studying
  the {\em grey zone} constraint~\cite{censor:2011,ghaffari:2013}:
  a natural geographic restriction on $G'$ that generalizes the unit disk graph model.
  Here we prove the perhaps surprising lower bound result
  that {\em every MMB algorithm}
  requires $\Omega((D+k)\fack)$ time, in the worst case,
  to solve the problem under this constraint.
  This result establishes the optimality of our analysis of BMMB under the grey zone
  constraint, as well as for arbitrary $G'$,
  and opens an intriguing gap between the grey zone and $r$-restricted assumptions.
  At the core of this lower bound is a careful scheduling strategy
 that synchronizes broadcasters in two parallel lines to a sufficient degree
 to allow their messages to cause mutual delay.

Having established the limits of MMB in the standard abstract MAC layer model,
we then explore the possibility of adding more power to the model to enable more efficient solutions.
In particular, we use the enhanced abstract MAC layer model of \cite{kuhn:2011abstract} which also allows nodes to abort
transmissions in progress and use knowledge of $\fprog$ and $\fack$.
Combining this model with the grey zone constraint on $G'$,
 we describe and analyze a new algorithm, which we call Fast Multi-Message Broadcast (FMMB),
that solves the MMB problem in $O((D\log{n} + k\log{n} + \log^3{n})\fprog)$ time (with high probability in the network size, $n$\footnote{We define high probability
to be at least $1-1/n$. We note that BMMB's guarantees are deterministic but that our lower bound works even for randomized solutions.})---avoiding an $\fack$ term altogether. This algorithm begins by building a maximal independent set (a subroutine of independent interest),
then efficiently gathers and spreads messages over the overlay network defined by these nodes.
We note that the assumptions that separate the enhanced from standard model were chosen in part
because they are feasible to implement using existing technology. 

\bigskip
\noindent {\bf Discussion.} From a theoretical perspective, the new upper and lower bounds proved in this
paper emphasize the perhaps surprising insight that the efficiency of message dissemination depends on the
{\em structure} of unreliability, not the {\em quantity.} We are able, for example, to solve MMB fast
with an $r$-restricted $G'$. This constraint allows for a large number of unreliable edges in every neighborhood,
and only forbids these edges from covering long distances in $G$. Our lower bound, on the other hand, demonstrates that even a small
number of unreliable edges is sufficient to slow down any MMB solution, so long as these edges are allowed to cover large
 distances in $G$.

  From a practical perspective,
  our efficient time bounds for  BMMB under the (reasonable) $r$-restricted assumption helps explain why
  straightforward flooding strategies tend to work well in real networks.
  In addition, our enhanced MAC layer results provide feedback to MAC layer designers, indicating
  that the ability to abort messages might prove crucial for enabling more efficient distributed protocols running on these layers.


\bigskip
\noindent {\bf Related Work.} The study of broadcast in wireless networks has a long line of history, dating back to 1985 work of Chalamatac and Kutten\cite{CK}, and has since received a vast amount of attention (see e.g., \cite{bar-yehuda:1992, alon:1991, KP, CR, KM93, KP-DC, GPX,KK, GHK13, Pelc08, DGKN10}). Most of this existing work deals directly with low-level issues such as managing contention on the shared medium.

The {\em abstract MAC layer} model was proposed in \cite{kuhn:2009,kuhn:2011abstract} as an alternative approach, which moves up the network
stack and abstracts away low level issues with model uncertainty. This model has since been used to study a variety of problems; e.g., ~\cite{khabbazian:2010, cornejo2009neighbor, cornejo2012reliable, khabbazian2011mac, Chung:2011, Chung:2010, Lynch:2012}.
Most relevant to this paper is the work of~\cite{kuhn:2009} and subsequently \cite{khabbazian:2010},
which study broadcast in various abstract MAC layer models, but under the assumption
of no unreliable edges.

A core property of the abstract MAC layer models studied in this paper, by contrast, is the presence
of unreliable links in addition to reliable links.
A lower level model also assuming these dual link types
 was introduced by Clementi et~al.~\cite{clementi:2004}, where it was called the {\em dynamic fault} model. 
 We independently reintroduced the model in~\cite{kuhn:2009} with the name {\em dual graph} model.
By now it is well-known that most problems that are simple in the absence of unreliability (when $G'=G$), become significantly harder in its presence (when $G'\neq G$); e.g., \cite{kuhn:2009, ghaffari:2012, ghaffari:2013, censor:2011}. For instance, in the dual graph model with an offline adaptive adversary, single-message broadcast require 
$\Omega(n)$ rounds, even in constant diameter graphs~\cite{kuhn:2009}. 
We emphasize, however, that this existing work on dual graphs focuses on low level models,
whereas this paper carries this behavior to a higher level of abstraction.

%



\section{Model and Problem}
\label{sec:model}

There is no single abstract MAC layer model, but instead many different such models
that all share the same strategy of abstracting standard wireless link layer behavior
and therefore preventing the algorithm from having to deal directly with low level wireless behavior.
Below we define the basics shared by these models, then specify
the two variants studied in this paper. We conclude by formalizing the multi-message broadcast problem.


\bigskip
\noindent {\bf Abstract MAC Layer Basics.}
Abstract MAC layer models typically define
the connectivity of a radio network with a pair of graphs, $G$ and $G'$,
where $G=(V,E)$, $G' = (V,E')$, and $E \subseteq E'$.
The $n$ vertices in $V$ correspond to the wireless devices (which we call {\em nodes} in this paper),
while the edges in $G$ and $G'$ describe
the communication topology.
At a high-level, edges in $E$ indicate reliable links over which the
model  can always deliver messages, while edges in $E'\setminus E$ indicate unreliable
links over which the model sometimes delivers messages and sometimes does not.

The model provides an acknowledged local broadcast primitive.
To simplify the definition of this primitive, assume
without loss of generality
that all local broadcast messages are unique.
When a node $u\in V$ {\em broadcasts} a message $m$, the model delivers the message to
all neighbors in $E$ and (perhaps) {\em some} neighbors in $E'\setminus E$.
It then returns an {\em acknowledgment} of $m$
to $u$ indicating the broadcast is complete.
These are the only message deliveries performed by the model.
%
We assume nodes are well-formed in the sense that
they always wait for the acknowledgment of their current message before initiating a new broadcast.

This model provides two timing bounds,
defined with respect to two positive constants, $\fack$ and $\fprog$ which are fixed for each execution.
The first is the {\em acknowledgment} bound, which guarantees that each broadcast will complete and be acknowledged
within $\fack$ time. The second is the {\em progress} bound, which guarantees the following slightly more complex condition:
fix some $(u,v)\in E$ and interval of length $\fprog$ throughout which $u$ is broadcasting a message $m$;
during this interval $v$ must receive some message (though not necessarily $m$).
%
The progress bound, in other words, bounds the time for a node to receive some message when at least one of its neighbors is broadcasting.
As mentioned in the introduction, in both theory and practice it is reasonable to assume that $\fprog$ is much smaller than $\fack$.
We emphasize that in abstract MAC layer models,
 the choice of which neighbors in $E' \setminus E$ receive a given message,
as well as the order of receive events, are determined non-deterministically by an arbitrary {\em message scheduler}.
The timing of these events is also determined non-deterministically by the scheduler, constrained only by the above
time bounds. 

We assume that nodes have unique ids.
We also assume that each node can
 differentiate between their neighbors in $E$ and $E' \setminus E$,
 an assumption justified by the standard practice in real networks of assessing link quality.
 %
For a given network definition $(G,G')$,
we use $D$ to describe the diameter of $G$,
and $d_G(u,v)$, for $u,v\in V$, to describe the shortest path distance between $u$ and $v$ in $G$.
We define $D'$ and $d_{G'}$ similarly, but for $G'$.
Finally, when proving lower bounds, we explicitly model randomness by passing each node at the beginning of the execution
sufficiently many random bits to resolve probabilistic choices.

\noindent {\bf The Standard Abstract MAC Layer.}
The standard abstract MAC layer  models nodes as event-driven automata.
It assumes that an environment abstraction fires a {\em wake-up} event at each node at the beginning of each execution.
The environment is also responsible for any events specific to the problem being solved. In multi-message broadcast,
for example, the environment provides the broadcast messages to nodes at the beginning of the execution.

\noindent {\bf The Enhanced Abstract MAC Layer.}
The enhanced abstract MAC layer model differs from the standard model in two
ways. First, it allows nodes access to time (formally, they can set timers that trigger events when they expire),
and assumes nodes know $\fack$ and $\fprog$.
Second, the model also provides nodes an {\em abort} interface that allows them to abort a broadcast
in progress.

\noindent {\bf Restrictions on $G'$.}
When studying a problem in an abstract MAC layer model,
it is often useful  to consider constraints on the graph $G'$.
In the original paper on these models~\cite{kuhn:2009},
for example, we considered the very strong constraint that $G'=G$.
In this paper, we consider three more general constraints on $G'$.

First, we say $G'$ is {\em arbitrary} to indicate that we place no restrictions on its definitions (other than the default constraint of $E \subset E'$).
Second,
we say $G'$ is
 {\em $r$-restricted}, for some $r \geq 1$,
if for every $(u,v)\in E'$, $d_G(u,v) \leq r$.
In studying this constraint, we sometimes use the notation $G^r$
to describe the graph with edges between every $u,v\in V$, $u\neq v$,
where $u$ and $v$ are within $r$ hops in $G$.
An $r$-restricted $G'$ is a subgraph of $G^r$.
Third, we say $G'$ is {\em grey zone restricted}
if (in addition to the general constraint of $E\subseteq E'$), the following is also true:
 we can embed the nodes in the Euclidean plane, giving each $v\in V$ a position $p(v) \in \mathbb{R}^2$
such that (1) For each pair of nodes $v, u \in V$, $(v, u) \in E$ if and only if $\left\|p(v)-p(u)\right\|_2 \leq 1$,
and (2) for each pair of nodes $v, u \in V$, if $(v, u) \in E'$, then $\left\|p(v)-p(u)\right\|_2 \leq c$,
where $c$ is a universal constant such that $c\geq 1$.
The range between $1$ and $c$, in other words, describes a {\em grey zone} in which
communication is uncertain. We emphasize that the second property described above only states that edges in $E'$ cannot be longer than $c$,
it does not require that every pair of nodes that have distance less than or equal to $c$ must be $G'$-neighbors. 

\noindent {\bf The Multi-Message Broadcast Problem.}
The multi-message broadcast (MMB) problem assumes the environment injects $k \geq 1$ messages into the network at the beginning of an 
execution,\footnote{A general
version of the MMB problem, in which the messages
arrive in an online manner, is studied in~\cite{kuhn:2011abstract} and elsewhere.}
perhaps providing multiple messages to the same node. We assume $k$ is not known in advance.
The problem is solved once every message $m$, starting at some node $u$,
reaches every node in $u$'s connected component in $G$.  To achieve strong upper bound we do not, in other words,
assume that $G$ is connected.
We treat messages as black boxes that cannot be combined; for example, we do not consider network coding solutions.
We also
assume that only a constant number of these messages can fit into a single local broadcast message.
In this paper, we consider both deterministic and randomized algorithms.
We require randomized solutions to solve the problem with high probability (w.h.p.), which
we define to be at least $1-1/n$.

\section{Multi-Message Broadcast with a Standard Abstract MAC Layer}
\label{sec:standard}

In this section we study multi-message broadcast in what we call the {\em standard abstract MAC layer model}.
%
As mentioned in the introduction,
 in previous work~\cite{kuhn:2009,kuhn:2011abstract}
we described the {\em Basic Multi-Message Broadcast} (BMMB) algorithm, which implements
the standard strategy of broadcasting each message only the first time you receive it.
In more detail, the BMMB protocol works as follows.
Every process $i$ maintains a FIFO queue and list of received messages.
When a process $i$ receives a message $m$ from the MAC layer
it checks whether it already received it.
If it has already received it, it discards the message. Otherwise, process $i$ adds $m$ to the back of its queue.
Process $i$ broadcasts the message at the head of its queue (if its queue is not empty) and waits for acknowledgment from the MAC layer.
When the MAC layer acknowledges the message,
$i$ removes it from the queue and moves on to the next message (if any).

In~\cite{kuhn:2011abstract}, we proved that BMMB solves the MMB problem in $O(D\fprog + k\fack)$ time
under the strong assumption that $G' = G$.
In the two subsections that follow, we study its behavior under more general definitions of $G'$.
We then prove a lower bound for all MMB solutions.

\subsection{The BMMB Algorithm for Arbitrary G'}
\label{sec:bmmb:arbitrary}

The natural next step in analyzing BMMB is considering its performance in our model when we assume an arbitrary $G'$.
It is easy to show, of course, that the algorithm always terminates in $O(Dk\fack)$ time, as
a message $m$, on arriving at a new hop, must wait for at most $O(k)$ messages to be broadcast before it too is broadcast.
Here we apply a more detailed pipeline argument to show that BMMB performs better than this basic bound in this difficult setting.

\begin{theorem}
 The BMMB algorithm solves the MMB problem in $O((D+k)\fack)$ time in the standard abstract MAC layer model for arbitrary $G'$.
\label{thm:bmmb:arbitrary}
\end{theorem}
\begin{proof}
For the sake of analysis, assume each node $u$ keeps
a $sent$ set to which it adds every message that it has successfully broadcast (i.e., broadcast and received an ack for).
Next, fix some execution and an arbitrary message $m$ from among the $k$ messages provided to BMMB to broadcast in this execution.
Let $u_m \in V$ be the node that is provided $m$ by the environment at the start of the execution.
For each node $v\in V$, let $d_v = d_G(v,u_m)$.
For each $\ell\in [1,k]$,
let $t_{\ell}(v) = d_v + \ell$.

Our strategy is to prove the following key claim:
{\em for each $\ell\in [1,k]$ and node $v$, after $t_{\ell}(v)\fack$ time, node $v$'s $sent$ set either: (1) contains
$m$; or (2) contains at least $\ell$ other messages.}
Once we prove this claim, it will then follow that after $t_k(v)\fack \leq (D + k)\fack$ time, $v$ has sent (and thus also received) $m$.
Applying this to all nodes and all $k$ messages then yields the theorem statement.
We prove our key claim using induction on $h = d_v + \ell$.
For the base case of $h=0$,
notice that $h=0$ implies $d_v=0$.
This, in turn, reduces the key claim
to a statement about the local queue of $v$ that follows directly from the definition of the algorithm.

For the inductive step,
consider some $v$ such that $d_v + \ell = h$.
To show the key claim holds for $t_{\ell}(v) =h$,
we leverage the inductive hypothesis for nearby nodes
and combinations of relevant values that sum to $h-1$.
First,
we note that
if $d_v=0$, then the base case argument once again applies.
Assume, therefore, that $d_v \geq 1$.
Next,
fix some $G$-neighbor $u$ of $v$ such that $d_u = d_v -1$.
By the induction hypothesis,
we know that after $s = t_{\ell-1}(v)\fack = t_{\ell}(u)\fack$ time:
$v$ has either sent $m$ or sent $\ell-1$ other messages,
and $u$ has either sent $m$ or sent $\ell$ other messages.
If $v$ has sent $m$ or at least $\ell$ messages by $s$ then we are done.
If it has not, then by time $s$,
$u$ will have either sent it $m$ or a new message (i.e., distinct from the $\ell-1$
messages $v$ has already sent by time $s$).
In either case, in the $\fack$ time that passes between
$s$ and $t_{\ell}(v)\fack$,
$v$ will either send $m$ or an $\ell^{th}$ message,
establishing the key claim for $t_{\ell}(v) = h$.
\end{proof}

\subsection{The BMMB Algorithm for r-Restricted G'}
\label{subsec:formal proof for BMMB}



We have shown that moving from $G'=G$ to an unrestricted
definition of $G'$ slows down the performance of BMMB; i.e., replacing a $D\fprog$
factor with $D\fack$, which might be substantially slower.
In this section we seek a middle ground---attempting to identify
just enough constraints on $G'$ to enable faster MMB performance.
In more detail, we study BMMB with an $r$-restricted $G'$ and prove that its performance
scales well with $r$:

\begin{theorem}
For any $r \geq 1$,
the BMMB algorithm solves the MMB problem in $O(D\fprog + r k\fack)$ time in the standard abstract MAC layer model
with an $r$-restricted $G'$.
\label{thm:bmmb:r}
\end{theorem}

Notice that for small $r$, this bound is better than the $O((D+k)\fack)$ bound we established in Theorem~\ref{thm:bmmb:arbitrary} for arbitrary $G'$,
and comes close to matching the bound proved in \cite{kuhn:2011abstract} for the case where $G' = G$.
This result implies the $D\fack$ factor of the arbitrary $G'$ bound is somehow dependent on the possibility
of $G'$ edges between nodes distant in $G$.
We emphasize that the above bound does not follow naturally from the $O(D\fprog + k \fack)$ bound
for $G'=G$. The presence of unreliability fundamentally disrupts the core induction of this existing bound
and requires a new approach.

\shortOnly{

The complete proof for Theorem~\ref{thm:bmmb:r} can be found in the full version of this paper~\cite{full-version}.
%
%
%
%
To provide intuition for this argument, however,
we present below a summary of the full proof and establish its key lemma. 

In more detail, fix some execution and an arbitrary message $\mm$, among the $k$ messages provided to $\BMMB$ to broadcast in this execution,
that arrive at some node $i_0\in V$ at time $t=0$. 
For use in our complexity analysis, we define certain time bounds.
For integers $\ell \geq 1 $ and $d\geq 0$, we define:
$$
t_{d,\ell}:=
(d+ (r+1)\ell-2)\cdot \fprog + r\cdot(\ell-1)\cdot \fack
~.$$
We also use the notation {\em completes}, with respect to a node $u$ and message $m$, to describe
the event of node $u$ receiving an acknowledgment for $m$.
The key lemma in our proof of Theorem~\ref{thm:bmmb:r} is:

\begin{lemma}\label{lemma:bmmblemma}
  Let $j$ be a node at distance $d=d_G(i_0,j)$ from $i_0$ in $G$.
  Let $\ell$ be any positive integer. Then:
	(1) Either $j$ receives $m$ or it receives at least $\ell$ messages by time $t_{d,\ell}$.
%
%
	(2) Either $j$ completes $m$ or it completes at least $\ell$ messages by time $t_{d,\ell}+\fack$.
%
%
\label{lem:bmmb:time}
\end{lemma}

Given that
$t_{D,k} =O(D \fprog + r\cdot k\cdot \fack)$,
Theorem \ref{thm:bmmb:r} follows directly from this lemma.
We note that
in the setting where $G'=G$, when  a node $j$ has completed sending a message, it follows that all of its neighbors must have received it.
In our setting, however, where  $G'\not= G$, we must also consider $G'$ neighbors for which this is not true.
In fact, although, one might expect that $G'$ neighbors will receive all message within $r\cdot \fack$ time,
that  too turns out to be not true, due to possible interference with other messages.
To cope with such problems,
we carry out a finer analysis of conditions that guarantee progress of messages through the network.
Specifically, we define two {\em positive progress} scenarios,
{\em fast positive progress}
and
{\em slow positive progress},
in which  node $i$ is guaranteed to receive some ``new'' message in $\fprog$ time and in $z\cdot \fack$ time, respectively (where $z$ is some positive integer).
(Recall that the basic progress definition does not require that the received message is ``new'': this is a condition we must established with respect
to our given algorithm.)
%
%
\begin{itemize}
\item {\bf \em Fast positive progress:}
%
Suppose that at time $t$,
(1) every message that $i$ has received is already completed at all of $i$'s
neighbors in $G'$; and
(2) some neighbor $j$ of $i$ in $G$ has received some message that $i$ has not yet received.
Then $i$ will receive a new message within $\fprog$ time.
\item 
{\bf \em Slow positive progress:}
Suppose that at time $t$,
(1) $j$ has completed at least $\ell-1$ messages;
(2) every message that $j$ has completed is also completed at all of $j$'s neighbors in $G^z$; and
(3) there exists at least one neighbor of $j$ in $G^z$ that has received at least $\ell$ messages.
Then $j$ will receive at least $\ell$ messages by time $t+ z\cdot \fack$.
\end{itemize}

With these preliminaries established, we can now sketch our proof for our key lemma.

\bigskip

{\em Proof Sketch} (for Lemma~\ref{lemma:bmmblemma}).
We use a double induction for $\ell$ as an ``outer'' induction and for distance $d$ as an ``inner'' induction.
Such a double induction was also used in \cite{kuhn:2009,kuhn:2011abstract} for the case where $G'=G$.
However, the current used different bounds and different inductive hypotheses, based on the two positive progress scenarios above.
For the base where $\ell=1$, we prove Part 1 by induction on $d$.
We assume that Part 1 holds for $d-1$.
Thus,
there exists a neighbor $j'$ of $j$ in $G$ at distance $d-1$ from $i_0$ (in $G$) such that $j'$ receives at least one message by time $t_{d-1,1}$.
Thus, $j$ receives at least one message by time $t_{d-1,1}+\fprog$.
Part 1 follows, since $t_{d,1}=t_{d-1,1}+\fprog$.
Hence, also, $j$ will complete its first message by time $t_{d,1} + \fack$ as needed for
Part 2.

Now, we are ready to show the inductive step for $\ell$.
Let $\ell\geq 2$.
We assume the lemma
statement for $\ell-1$ (and for all $d$) and prove it for $\ell$.
To prove the lemma statement for $\ell$, we use a second, ``inner'' induction, on the distance
$d$ (in $G$) from $i_0$  and the destination $j$.
For the base, $d = 0$, Part 1 follows immediately.
Part 2 follows, since
(1) $
\ell \fack \leq t_{\ell,0}$;
(2) if the position of $m$ in $bcastq_{i_0}$ is $\leq \ell$, then $i_0$ completes $m$ by time
$\ell \fack$; and
(3) if the position of $m$ in $bcastq_{i_0}$ is
greater than $\ell$, then $i_0$ completes at least $\ell$ messages by time
$\ell \fack$.

For the inductive step of the inner induction, we assume $d \geq 1$.
Assume both parts of the lemma for
(1) $\ell -1$ (as ``outer'' induction hypothesis) for all distances;
and
(2) for $\ell$ for distance $d-1$ (as ``inner'' induction hypothesis).
We prove both parts of the lemma for $\ell$ and distance $d$.

We now sketch Part 1.
(Part 2 follows from Part 1.)
Let $t^*=t_{d+r,\ell-1}$.
Note that $t^* +r\cdot \fack \leq t_{d,\ell}$ (by a simple algebraic calculation).
If either (1) $j$ completes $m$ by time $t^*$ or
(2) $j$ completes at least $\ell$ messages by time $t^*$,
then Part 1 follows, since $t^*<t_{d,\ell}$.
So, suppose the contrary, that $j$ does not completes $m$ and it completes exactly $\ell-1$ messages by time $t^*$.
(By the outer inductive hypothesis for $\ell-1$, if $j$ does not complete $m$ by time $t^*$, then it must complete at least $\ell-1$ messages by that time.)
We then show
(by using the outer inductive hypothesis for $\ell-1$ for all neighbors $j'$ of $j$ in $G^r$),
that, if there exists some neighbor $j'$ of $j$ in $G^r$ that
completes, by time $t^*$, a different message set than the message set that $j$ completes by that time,
then the {\em slow positive progress scenario} occurs at time $t^*$ with respect to some $z\leq d_{G}(j,j')\leq r$.
This implies that $j$ receives at least $\ell$ messages by time $t^*+r\cdot \fack\leq t_{d,\ell}$, as needed for Part 1.

Finally, we consider the case where $j$ and all of its neighbors in $G^r$ complete exactly the same message set by time $t^*$.
In that case, hypothesis (1) of the fast positive progress scenario holds at time $t^*$, but  hypothesis (2) of that scenario does not necessarily holds for that time.
We then use the outer  inductive hypothesis for $\ell$ with the inner inductive hypothesis for $d-1$ to show that there exists a neighbor $j'$ of $j$ in $G$ at distance $d-1$ from $i_0$
(in $G$) that receives at least $\ell$ messages by time $t_{d-1,\ell}$.
This implies that, if $j$ receives exactly $\ell-1$ messages by time $t_{d-1,\ell}$ (that is, $j$ does not receive additional messages during the time interval $[t^*,t_{d-1,\ell}]$),
then the {\em fast positive progress scenario} occurs at time $t_{d-1,\ell}$.
Therefore, in both cases, $j$ receives at least $\ell$ messages by time $t_{d-1,\ell} + \fprog$.
Part 1 follows since $t_{d,\ell}=t_{d-1,\ell}+ \fprog$ (by a simple algebraic calculation).
\qed

} 

\fullOnly{


We present the 
proof of  Theorem~\ref{thm:bmmb:r} using a more formal description of the abstract MAC layer.
We can model our systems using \emph{Timed I/O Automata} \cite{TIOA2010}.

Consider some positive integer $r$.
The $r$'th power $G^r(V,E^r)$ of a graph $G$ is a graph with the same vertex set $V$, and two nodes $v,u\in V$
are adjacent, $(u,v)\in E^r$, when their distance in $G$ is at most $r$.
That is,
$E^r=\{ (u,v) \mid u\not= v \mbox{ and } d_G(u,v) \leq r\}$.
(The $r$th power  graph $G^r$ does not includes self-loops.)
For a given node $j\in V$,
let $N^r_G(j)=\{ j' \mid d_G(j,j')\leq z \}$ denote the set of nodes that are within $r$ hops of $j$
in $G$, including $j$ itself.
A graph $G'=(V',E')$ is a subgraph of $G=(V,E)$ (denoted by $G\subseteq G'$), if $V=V'$ and $E' \subseteq E$.

\subsubsection{Guarantees for the Abstract MAC Layer}
\label{sec:model:prop}

Here we provide a set of properties that constrain the
behavior of the abstract MAC layer automaton.
Technically, these properties are expressed for admissible timed
executions of the timed I/O automaton modeling the complete system.

\paragraph{Well-Formedness Properties.}
We assume some constraints on the behavior of the user automata.
Let $\alpha$ be an admissible execution of the system consisting of
user and abstract MAC layer automata.
We say $\alpha$ is {\em user-well-formed} if the following hold, for
every $i$:
\begin{enumerate}
\item
Every two $bcast_i$ events have an intervening $ack_i$ or $abort_i$ event.
\item
Every $abort(m)_i$ is preceded by a $bcast(m)_i$ (for the same $m$)
with no intervening $bcast_i$, $ack_i$, or $abort_i$ events.
\end{enumerate}

The rest of this subsection gives constraints on the behavior of the
abstract MAC layer automaton, in system executions that are
user-well-formed.
Thus, from now on in the section, we assume that $\alpha$ is a
user-well-formed system execution.

\paragraph{Constraints on Message Behavior.}
We assume that there exists a ``cause'' function that maps every $rcv(m)_j$
event in $\alpha$ to a preceding $bcast(m)_i$ event, where $i \ne j$,
and that maps each $ack(m)_i$ and $abort(m)_i$ to a preceding $bcast(m)_i$.

We use the term \emph{message instance} to refer to the set
consisting of a $bcast$ event and all the other events that are
related to it by the \emph{cause} function.

We now define two safety conditions and one liveness condition
regarding the relationships captured by the cause function:

\begin{enumerate}
\item {\bf Receive correctness:}
Suppose that $bcast(m)_i$ event $\pi$ causes $rcv(m)_j$ event $\pi'$
in $\alpha$.
Then $(i,j) \in E'$, and no other $rcv(m)_j$ event or $ack(m)_i$ event
caused by $\pi$ precedes $\pi'$.

\item  {\bf Acknowledgment correctness:}
Suppose that $bcast(m)_i$ event $\pi$ causes $ack(m)_i$ event $\pi'$ in $\alpha$.
Then for every $j$ such that $(i,j) \in E$, a $rcv(m)_j$ event caused
by $\pi$ precedes $\pi'$.
Also, no other $ack(m)_i$ event or $abort(m)_i$ caused by $\pi$
precedes $\pi'$.

\item  {\bf Termination:}
Every $bcast(m)_i$ causes either an $ack(m)_i$ or an $abort(m)_i$.
\end{enumerate}

\paragraph{Time Bounds.}
We now impose upper bounds on the time from a $bcast(m)_i$ event to
its corresponding $ack(m)_i$ and $rcv(m)_j$ events.\footnote{
  We express these here as constants rather than functions, because we
  will not worry about adaptive bounds in this paper.
  }

Let $\fack$ and $\fprog$ be positive real numbers.
We use these to bound delays for a specific message to be delivered
and an acknowledgment received (the ``acknowledgment delay''), and for
{\em some} message from among many to be received (the ``progress
delay'').
We think of $\fprog$ as smaller than $\fack$, because the time to
receive {\em some} message among many is typically less than the time
to receive a specific message.
The statement of the acknowledgment bound is simple:

\begin{enumerate}
\setcounter{enumi}{3}
\item
{\bf Acknowledgment bound:}
Suppose that a $bcast(m)_i$ event $\pi$ causes an $ack(m)_i$ event
$\pi'$ in $\alpha$.
Then the time between $\pi$ and $\pi'$ is at most $\fack$.
\end{enumerate}

The statement of the progress bound is a bit more involved, and
requires some auxiliary definitions.
Let $\alpha'$ be a closed execution fragment within the given execution
$\alpha$,\footnote{Formally, that means that there
  exist fragments $\alpha''$ and $\alpha'''$
  such that $\alpha = \alpha'' \alpha' \alpha'''$, and moreover, the first
  state of $\alpha'$ is the last state of $\alpha''$. Notice,
  this allows $\alpha'$ to begin and/or end in the middle of a trajectory.}
and let $j$ be a process.
Then define:
\begin{itemize}
\item
$connect(\alpha',j)$ is the set of message instances in $\alpha$ such that
$\alpha'$ is wholly contained between the $bcast$ and terminating
event ($ack$ or $abort$) of the instance, and $(i,j) \in E$, where $i$
is the originator of the $bcast$ of the message instance.
\item
$contend(\alpha',j)$ is the set of message instances in $\alpha$ for
which the terminating event does not precede the beginning of $\alpha'$, and
$(i,j) \in E'$, where $i$ is the originator of the $bcast$ of the
message instance.
\end{itemize}

\begin{lemma}
For every $\alpha'$ and $j$,
$connect(\alpha', j) \subseteq contend(\alpha',j)$.
\end{lemma}

\begin{enumerate}
\setcounter{enumi}{4}
\item
{\bf Progress bound:}
For every closed fragment $\alpha'$ within $\alpha$, and for every
process $j$, it is not the case that all three of the following
conditions hold:
  \begin{enumerate}
  \item The total time of $\alpha'$ is strictly greater than $\fprog$.
  \item $connect(\alpha', j) \neq \emptyset$.
  \item No $rcv_j$ event from a message instance in $contend(\alpha', j)$
    occurs by the end of $\alpha'$.
  \end{enumerate}
In other words, $j$ should receive {\em some} message within time
$\fprog$ provided that at least one message is being sent by a $G$-neighbor.
\end{enumerate}

Note that our definitions allow a $rcv$ for a particular $bcast$
to occur after an $abort$ for that $bcast$.
We impose a (small) bound $\epsilon_{abort}$ on the amount of time
after an $abort$ when such a $rcv$ may occur.

%

\subsubsection{The Multi-Message Broadcast Problem}
\label{sec:mmb:prelim}

A user automaton is considered to be an {\em MMB protocol} provided
that its external interface includes an $arrive(m)_i$ input
and $deliver(m)_i$ output for each user process $i$
and message $m\in {\cal M}$.

We say an execution of an MMB protocol is {\em MMB-well-formed} if and only if it
contains at most one $arrive(m)_i$
event for each $m\in {\cal M}$;
that is, each broadcast message is unique.
We say an MMB protocol {\em solves the MMB problem}
if and only if for every MMB-well-formed (admissible) execution $\alpha$
of the MMB protocol composed with a MAC layer,
the following hold:
\begin{enumerate}
\item[(a)]
For every $arrive(m)_i$ event in $\alpha$ and every process $j$,
$\alpha$ contains a $deliver(m)_j$ event.
\item[(b)]
For every $m\in {\cal M}$ and every process $j$,
$\alpha$ contains at most one $deliver(m)_j$ event and it
comes after an $arrive(m)_i$ event for some $i$.
\end{enumerate}

We describe a simple MMB protocol that achieves efficient runtime.

\begin{quote}\small
{\bf The Basic Multi-Message Broadcast (BMMB) Protocol }\\
Every process $i$ maintains a FIFO queue named $bcastq$ and a set
named $rcvd$.
Both are initially empty.

If process $i$ is not currently sending a message (i.e., not waiting
for an $ack$ from the MAC layer) and $bcastq$ is not empty, the
process immediately (without any time-passage) $bcast$s the message at
the head of $bcastq$ on the MAC layer.

When process $i$ receives an $arrive(m)_i$ event, it immediately
performs a local $deliver(m)_i$ output and adds $m$ to the back of its
$bcastq$, and to its $rcvd$ set.

When $i$ receives a  message $m$ from the MAC layer it checks
its $rcvd$ set.
If $m \in rcvd$, procsss $i$ discards the message.
Otherwise, $i$ immediately performs a $deliver(m)_i$ event,
and adds $m$ to the back of its $bcastq$ and to its $rcvd$ set.
\end{quote}

\begin{theorem}
The BMMB protocol solves the MMB problem.
\label{thm:bmmb}
\end{theorem}

We give two definitions that we will use in our complexity proof.
In the following, let $\alpha$ be some MMB-well-formed
execution of the BMMB protocol composed with a MAC layer.
We begin with two definitions that we will use in our complexity proof.

\paragraph{$get$ events.}
We define a $get(m)_i$ event with respect to $\alpha$,
for some arbitrary message $m$ and process $i$,
to be one in which process $i$ first learns about message $m$.
Specifically, $get(m)_i$ is the first $arrive(m)_i$ event
in case message $m$ arrives at
process $i$,
otherwise, $get(m)_i$ is the first $rcv(m)_i$ event.

\paragraph{$clear$ events}
Let $m\in\mathcal{M}$ be a message for which an $arrive(m)_i$ event
occurs in $\alpha$.  We define $clear(m)$ to describe the final
$ack(m)_j$ event in $\alpha$ for any process $j$.\footnote{
  By the definition of BMMB, if an $arrive(m)_i$ occurs, then $i$
  eventually $bcast$s $m$, so $ack(m)_i$ eventually occurs. Furthermore,
  by the definition of BMMB, there can be at most one $ack(m)_j$ event for
  every process $j$. Therefore, $clear(m)$ is well-defined.}

\subsubsection{Proof Preliminaries}
\label{subsec:Proof Prelim}

For the rest of Section \ref{sec:standard}, we consider the special case of the general MMB problem
in which all messages arrive from the environment at time $t_0 = 0$, that is, all {\em arrive} events occur at time 0.
We fix a particular MMB-well-formed execution $\alpha$ of the BMMB protocol composed with a MAC layer.
We assume that an $arrive(m)_{i_0}$ event occurs in $\alpha$ for
some message $m \in \mathcal{M}$ at some node $i_0$, at time $t_0=0$.

For each node $i\in V$ and each time $t$, we introduce two sets of messages,
which we call $\cR$ (for "received messages") and $\cC$ (for "completed messages").
%
%
%
%
%
We define:
\cmpdell{\footnote{So $\alpha$, $m$, $i_0$, 
  and $\mathcal{M}'$ are implicit parameters.}}
\begin{itemize}
\item
$\cR_i(t) \subseteq \mathcal{M}\cmpdell{'}$ is the set of messages $m' \in
\mathcal{M}\cmpdell{'}$ for which the  $get(m')_i$ event occurs by time $t$.
\item
$\cC_i(t)\subseteq\mathcal{M}\cmpdell{'}$ is the set of messages
$m'\in\mathcal{M}\cmpdell{'}$ for which the $ack(m')_i$ event occurs by time
$t$.
\end{itemize}
That is, $\cR_i(t)$ is the set of messages that have been received by
process $i$ by time $t$, and $\cC_i(t)$ is the set of messages that
process $i$ has finished (completed) processing by time $t$.\footnote{
Note that both $\cR_i(t)$ and $\cC_i(t)$ may include $m$.}

The following two lemmas express some very basic properties of the $\cR$ and $\cC$ sets.

\begin{lemma}

For every $i,i',t$ and $t'$ such that $t\leq t'$:
  \begin{enumerate}
  \item $\cR_i(t)\subseteq \cR_i(t')$ and $\cC_i(t)\subseteq \cC_i(t')$.
  \label{item:lemma: t<t' => C(t) <= C(t')}

  \item $\cC_{i}(t) \subseteq \cR_{i}(t')$.
  \label{item:lemma: C_i(t) subseteq R_i(t)}

  \item If $i$ and $i'$ are neighbors in $G$, then $\cC_{i'}(t) \subseteq \cR_{i}(t')$.
  \label{item:lemma: i and j neigh t'<=t C_i(t') subseteq R_j(t)}

  \end{enumerate}
\label{lemma:bmmb: C(t') subset C(t'')}
\label{lemma: CG1}
\label{lemma: CG4}
\end{lemma}
\begin{proof}
Straightforward.
\end{proof}

\begin{lemma}
\label{lemma: CG6}
Fix $i$ and $t \geq 0$, and let $s$ be the final state at time $t$.
Then, in state $s$, $\bcastq_i$ contains exactly the messages in $\cR_i(t) - \cC_i(t)$.
\cmpdell{
\begin{enumerate}
\item
If in state $s$, $m'$ is in $\bcastq_i$, then $m' \in \mathcal{M}'$.
\item
In state $s$, $\bcastq_i$ contains exactly the messages in $\cR_i(t) - \cC_i(t)$.
\end{enumerate}
}
\end{lemma}

\begin{proof}
By the operation of the algorithm, $\bcastq_i$ contains exactly the
messages that have had a $get$ and no $ack$ at node $i$.
These
are exactly the elements of $\cR_i(t) - \cC_i(t)$.
\end{proof}

\begin{lemma}
\label{lemma:  CG2}
\label{lemma: CG3}
Fix $i$ and $t \geq 0$,
and let $s$ be the final state at time $t$.
\begin{enumerate}
\item
If in state $s$, $m'$ is in position $k\geq 1$ of $bcastq_i$, then $m'
\in \cC_i(t + k \fack)$.
\item
If in state $s$, $bcastq_i$ has length at least $k\geq 0$, then
$|\cC_i(t + k \fack)| \geq |\cC_i(t)| + k$.
\item
If $|\cR_i(t)| \geq k \geq 0$, then $|\cC_i(t + k \fack)| \geq k$.
\end{enumerate}
\end{lemma}

\begin{proof}
\begin{enumerate}
\item
For $k=1$,
the abstract MAC layer properties say that, within
time $\fack$, $m'$ is acknowledged at $i$.
Therefore, $m' \in \cC_i(t + \fack)$, which yields the result for $k=1$.
The statement for general $k$ follows from repeated application of the
statement for $k=1$.
\item
The statement is trivial for $k=0$, so consider $k \geq 1$.
Part 1, applied to the first $k$ messages in $bcastq_i$ in state $s$,
implies that all of these messages are in $\cC_i(t + \fack)$.
Lemma~\ref{lemma: CG6},
implies that none of these messages
are in $\cC_i(t)$.
Therefore, $|\cC_i(t + \fack)| \geq |\cC_i(t)| + k$.
\item
Suppose that  $|\cR_i(t)| = |\cC_i(t)| + |\cR_i(t) - \cC_i(t)| \geq k$.
By Lemma~\ref{lemma: CG6}, every element of $\cR_i(t) - \cC_i(t)$ is on
$bcastq_i$ in state $s$, so the length of $bcastq_i$ in state $s$ is
at least $|\cR_i(t) - \cC_i(t)|$.
We consider two cases. If $|\cR_i(t) - \cC_i(t)|\leq k$, then
$$
|\cC_i(t+k \cdot \fack)| \geq |\cC_i(t+ |\cR_i(t) - \cC_i(t)| \cdot \fack)| \geq |\cC_i(t)|+|\cR_i(t) - \cC_i(t)| \geq k,
$$
as needed for Part 3, where the first inequality follows, since
$\cC_i(t+ |\cR_i(t) - \cC_i(t)|\cdot \fack) \subseteq \cC_i(t+k \cdot \fack) $
by
Part \ref{item:lemma: t<t' => C(t) <= C(t')}
of Lemma \ref{lemma:bmmb: C(t') subset C(t'')}
(with $t= t+ |\cR_i(t) - \cC_i(t)|\cdot \fack$ and $t'=t+k \cdot \fack$);
the second inequality follows by Part 2 of the lemma;
and the last inequality holds by the assumption of Part 3 (of the lemma).
If $|\cR_i(t) - \cC_i(t)| > k$, then
$|\cC_i(t+k \cdot \fack)| \geq |\cC_i(t)|+k$, by Part 2 of the lemma.
Part 3 follows.
\end{enumerate}
\end{proof}

The following corollary is a special case of Part 2 of Lemma \ref{lemma:  CG2}.

\begin{corollary}
\label{lemma:bmmb: ack and prog progress}
\label{coro:bmmb: ack and prog progress}
Fix $i$, $t\geq 0$ and $\ell>0$.
If $|\cC_i(t)| \geq \ell-1$ and $|\cR_i(t)| \geq \ell$, then $|\cC_i(t+\fack)| \geq \ell$.
\end{corollary}
\begin{proof}
The case where $|\cC_i(t)|\geq \ell$ follows
by Part \ref{item:lemma: t<t' => C(t) <= C(t')}
of Lemma \ref{lemma:bmmb: C(t') subset C(t'')}.
Suppose that $|\cC_i(t)| = \ell-1$.
Since $|\cR_i(t)| \geq \ell$, it follows that $\cR_i(t)-\cC_i(t)\not=\emptyset$.
Therefore,
by
Lemma \ref{lemma: CG6},
at the final state $s$ at time $t$, $bcastq_i$ has length at least at least one.
Thus, by Part 2 of Lemma \ref{lemma:  CG2}, $|\cR_i(t+\fack)| \geq |\cR_i(t)|+1$.
The Corollary follows.
\end{proof}

Now we have two key lemmas that describe situations when a process $i$
is guaranteed to receive a new message.
The first deals with $i$ receiving its first message.\footnote{
  Actually, this lemma is formally a corollary to the following one,
  but it might be nicer to see this proof as a ``warm-up''.
  }

\begin{lemma}
\label{lemma: CG5}
Let $i$ and $j$ be neighboring nodes in $G$, and suppose $t \geq 0$.
If $\cR_j(t) \neq \emptyset$, then $\cR_i(t + \fprog) \neq \emptyset$.
\end{lemma}

\begin{proof}
Assume for contradiction that $\cR_i(t + \fprog) = \emptyset$.
Choose $t' > t + \fprog$ to be some time strictly after $t + \fprog$,
when $\cR_i(t') = \emptyset$;
this is possible because the next discrete event after time $t +
\fprog$ must occur some positive amount of time after $t + \fprog$.

We obtain a contradiction to the progress bound.
Let $\alpha'$ be the closed execution fragment of $\alpha$ that begins
with the final state $s$ at time $t$ and ends with the final state
$s'$ at time $t'$.
We show that $\alpha'$ provides a contradiction to the progress
bound.  We verify that the three conditions in the definition of the
progress bound are all satisfied for $\alpha'$.  Condition (a), that
the total time of $\alpha'$ is strictly greater than $\fprog$,
is immediate.

Condition (b) says that $connect(\alpha',i) \neq \emptyset$.
Since $\cR_i(t) = \emptyset$, Lemma~\ref{lemma: CG4}, Part 3, implies
that $\cC_j(t) = \emptyset$.
Since $\cR_j(t) \neq \emptyset$, Lemma~\ref{lemma: CG6},
implies
that, in state $s$, $bcastq_j$ is nonempty.
Let $m'$ be the message at the head of $bcastq_j$ in state $s$.
Since $s$ is the final state at time $t$ and the protocol has $0$
delay for performing $bcast$s, it must be that the $bcast$ event for
process $j$'s instance for $m'$ occurs before the start of $\alpha'$.
Since $m' \notin \cR_i(t')$, $m'$ is not received by process $i$ by the
end of $\alpha'$.
This implies that the $ack_j(m')$ event, which terminates $j$'s
instance for $m'$, must occur after the end of $\alpha'$.
It follows that $j$'s instance for $m'$ is in $connect(\alpha', i)$,
so that $connect(\alpha',i) \neq \emptyset$, which shows Condition
(b).

Condition (c) says that no $rcv_i$ event from a message instance in
$contend(\alpha',i)$ occurs by the end of $\alpha'$.
%
%
%
%
We know that no $rcv_i$ occurs by the end of $\alpha'$, because $\cR_i(t') = \emptyset$.
So no $rcv_i$ event from a message instance in $contend(\alpha',i)$
occurs by the end of $\alpha'$, which shows Condition (c).

Thus, $\alpha'$ satisfies the combination of three conditions that
are prohibited by the progress bound assumption, yielding the needed
contradiction.
\end{proof}

The next lemma deals with the fast positive progress scenario in which  process $j$ receives some ``new'' message in $\fprog$ time.
\begin{lemma}
Let $i$ and $j$ be neighboring nodes in $G$, and suppose $t \geq 0$. Suppose that:
\begin{enumerate}
\item $\cR_i(t) \subseteq \cC_{i'}(t)$ for every neighbor $i'$ of $i$ in $G'$.
\item $\cR_j(t) - \cR_i(t) \not= \emptyset$.
\end{enumerate}

Then,  $|\cR_i(t + \fprog) | > |\cR_i(t)|$.
\label{lemma: z-neighborhood prog bound}
\end{lemma}
This says that, if every message that $i$ has already received is already completed at all of $i$'s
neighbors in $G'$ and some neighbor $j$ of $i$ in $G$ has received some message that $i$ hasn't yet received, then $i$ will
receive a new message within $\fprog$ time.

\begin{proof}
Assume for contradiction that
$\cR_i(t) \subseteq \cC_{i'}(t)$ for every neighbor $i'$ of $i$ in $G'$,
that $\cR_j(t) - \cR_i(t) \neq \emptyset$, and
that $|\cR_i(t+\fprog)| = |\cR_i(t)|$.
Then it must be that $\cR_i(t+\fprog) = \cR_i(t)$.
Choose $t' > t + \fprog$ to be some time strictly after $t + \fprog$,
when $\cR_i(t') = \cR_i(t)$; this is possible because the next
discrete event after time $t + \fprog$ must occur some positive
amount of time after $t + \fprog$.

We obtain a contradiction to the progress bound.
Let $\alpha'$ be the closed execution fragment of $\alpha$ that begins
with the final state $s$ at time $t$ and ends with the final state
$s'$ at time $t'$.  We verify that the three conditions in the
definition of the progress bound are all satisfied for $\alpha'$.
\begin{itemize}
\item
Condition (a):  The total time of $\alpha'$ is strictly greater
than $\fprog$.  \\
This is immediate.

\item
Condition (b):  $connect(\alpha',i) \neq \emptyset$. \\
Since $\cR_j(t) - \cR_i(t) \neq \emptyset$ and $\cC_j(t) \subseteq \cR_i(t)$
(by Part
\ref{item:lemma: i and j neigh t'<=t C_i(t') subseteq R_j(t)}
of Lemma \ref{lemma:bmmb: C(t') subset C(t'')} with $i=i$, $i'=j$),
we have
$\cR_j(t) - \cC_j(t) \neq \emptyset$.

Then Lemma~\ref{lemma: CG6},
implies that, in state $s$,
$\bcastq_j$ is nonempty.
Let $m'$ be the message at the head of $\bcastq_j$ in state $s$.
Since $s$ is the final state at time $t$ and the protocol has $0$
delay for performing $bcast$s, it must be that the $bcast$ event for
process $j$'s instance for $m'$ occurs before the start of $\alpha'$.

Also, we know that $m' \notin \cR_i(t)$, because $m' \notin \cC_j(t)$ and
$\cR_i(t) \subseteq \cC_j(t)$.

Since $\cR_i(t') = \cR_i(t)$, we also know that $m' \notin \cR_i(t')$.
Therefore, $m'$ is not received by process $i$ by the end of
$\alpha'$.
This implies that the $ack_j(m')$ event, which terminates $j$'s
instance for $m'$, must occur after the end of $\alpha'$.
It follows that $j$'s instance for $m'$ is in $connect(\alpha', i)$,
so that $connect(\alpha',i) \neq \emptyset$.

\item
Condition (c):  No $rcv_i$ event from a message instance in
$contend(\alpha',i)$ occurs by the end of $\alpha'$.
%
%
%
We claim that, if a message $m''$ has an instance in
$contend(\alpha',i)$, then $m'' \notin \cR_i(t)$.
To see this, let $i'$ be a neighbor of $i$ in $G'$ originating an instance of
$m''$ in $contend(\alpha',i)$.
If $m'' \in \cR_i(t)$, then by hypothesis 1 (of the lemma), also $m'' \in \cC_{i'}(t)$.
That means that the $ack$ event of node $i'$'s instance of $m''$
occurs before the start of $\alpha'$, which implies that the instance
is not in $contend(\alpha,i)$.


With this claim, we can complete the proof for Condition (c).
Suppose for contradiction that a $rcv_i(m'')$ event from some message
instance in $contend(\alpha,i)$ occurs by the end of $\alpha'$.
Using the first claim above, $m'' \in \cR_i(t')$.
But by the second claim above, $m'' \notin \cR_i(t)$.
But we have assumed that $\cR_i(t') = \cR_i(t)$, which yields a
contradiction.
\end{itemize}

Thus, $\alpha'$ satisfies the combination of three conditions that
are prohibited by the progress bound assumption, yielding the needed
contradiction.
\end{proof}

The next lemma deals with the slow positive progress scenario in which  process $j$ is guaranteed to receive some ``new'' message in $z \cdot \fack$ time (for some positive integer $z$).

\begin{lemma}
Fix some time $t\geq 0$. Suppose that:
\begin{enumerate}

\item
$|\cC_{j}(t)|\geq \ell-1$.

\item
$\cC_{j}(t) \subseteq \cC_{j'}(t)$ for every node $j' \in N_G^z(j)$.

\item
There exists some $j'' \in N_G^z(j)$ such that $|\cR_{j''}(t) |\geq \ell$.
\label{item: condition: there exists some j'' s t Rj''(t) geq ell}

\end{enumerate}
Then $|\cR_j(t + z \fack)| \geq \ell$.
\label{lemma: progress in z cdot Fack}
\end{lemma}

\noindent This says that, if
(1) $j$ completes at least $\ell-1$ messages by time $t$;
(2) every message that $j$ has completed by time $t$ is also completed at all of $j$'s
neighbors in $G^z$ by time $t$; and
(3) there exists at least one neighbor of $j$ in $G^z$ that receives at least $\ell$ messages by that time,
then $j$ receives
at least $\ell$ messages by time $t+ z\cdot \fack$~.

\begin{proof}
We prove this lemma by induction on $z$.
For the base, $z= 0$, the statement trivially follows, since $N^0(j)=\{j\}$, which implies together with condition
\ref{item: condition: there exists some j'' s t Rj''(t) geq ell} that $|\cR_{j}(t) |\geq \ell$.
For the inductive step, we assume $z \geq 1$.
We assume the lemma statement for $z'<z$ and prove it for $z$.
If $|\cR_j(t)| \geq \ell$, then
by Part \ref{item:lemma: t<t' => C(t) <= C(t')}
of Lemma \ref{lemma:bmmb: C(t') subset C(t'')},
$|\cR_j(t+\fack)| \geq \ell$ and we are done.

If $\cC_{j}(t+\fack) \not= \cC_{j}(t)$, then by  assumption 1
and Part
\ref{item:lemma: t<t' => C(t) <= C(t')}
of Lemma \ref{lemma:bmmb: C(t') subset C(t'')},
$|\cC_{j}(t+\fack)|\geq \ell$,
which implies
by Part \ref{item:lemma: C_i(t) subseteq R_i(t)}
of Lemma \ref{lemma:bmmb: C(t') subset C(t'')}
that $|\cR_{j}(t+z\cdot \fack)|\geq \ell$, as needed.

It remains to consider the case where $|\cR_j(t)| = \ell-1$ and  $\cC_{j}(t+\fack) = \cC_{j}(t)$.
Let $j''$ be a closest neighbor of $j$ such that $|\cR_{j''}(t)| \geq \ell$.
Note that, $j''$ must be at distance at least 1 from $j$ (by the assumptions for this case) and $j''$ must be at distance at most $z$ from $j$ (by assumption 3 of the lemma).
Moreover, by the first two assumptions of the lemma, it follow that $|\cC_{j''}(t)| \geq \ell-1$.
Combining this inequality with  $|\cR_{j''}(t)| \geq \ell$ and
Corollary \ref{coro:bmmb: ack and prog progress},
we get that
\begin{equation}
|\cC_{j''}(t+\fack)| \geq \ell.
\label{ineq:lemma:|C_j''(t + Ffac)| geq ell}
\end{equation}

Let $j^*$ be the next-closer node on a shortest path in $G$ from $j''$ to $j$.
We now apply the inductive hypothesis
for $z' = z-1$ and time $t'=t+\fack$.
Note that, $j^*\in N^{z'}_G(j)$.
To do this, we show that the three assumptions of the lemma indeed hold for $z'$ and $t'$.
The first assumption of the lemma holds as
$$|\cC_{j}(t')|\geq |\cC_{j}(t)|\geq \ell-1,$$
where the first inequality holds since
$\cC_{j}(t) \subseteq \cC_{j}(t')$
(by Part \ref{item:lemma: t<t' => C(t) <= C(t')}
of Lemma \ref{lemma:bmmb: C(t') subset C(t'')},
with $t\leq t'$);
and the second inequality holds by assumption 1 for $z$ and $t$.
The second assumption of the lemma holds as,
$$
\cC_{j}(t') \subseteq \cC_{j'}(t) \subseteq \cC_{j'}(t'), \mbox{ for every node } j' \in N_G^{z'},
$$
where the first inequality holds since, in this case, $\cC_{j}(t') = \cC_{j}(t)$ and $\cC_{j}(t) \subseteq \cC_{j'}(t)$ (by assumption 2 of the lemma for $z>z'$ and $t$);
and the second inequality holds by Part \ref{item:lemma: t<t' => C(t) <= C(t')}
of Lemma \ref{lemma:bmmb: C(t') subset C(t'')} with $t\leq t'$.
We next argue that the third assumption holds as well.
Specifically, we claim that $j^*$, in particular, satisfies this assumption.
That is,
$j^*\in N_G^{z'}(j)$ and $|\cR_{j^*}(t')| \geq \ell$.
The first statement holds, since $z'=z-1$ and $d_G(j,j^*) < d_G(j,j'')\leq z$.
The second statement holds as,
$$
|\cR_{j^*}(t')| \geq |\cC_{j''}(t')|\geq \ell,
$$
where the first inequality holds,
since $\cC_{j''}(t') \subseteq \cR_{j^*}(t')$
(by Part \ref{item:lemma: i and j neigh t'<=t C_i(t') subseteq R_j(t)}
of Lemma \ref{lemma:bmmb: C(t') subset C(t'')})
and the second inequality holds
by combining together Inequality (\ref{ineq:lemma:|C_j''(t + Ffac)| geq ell}) with $t'=t+\fack$.

Having shown the three assumptions, we can now invoke the inductive hypothesis for $z'$ and $t'$.
We have $|\cR_{j}(t'+z'\cdot \fack)| \geq \ell$.
In addition, $t'+z'\cdot \fack \leq t+z\cdot \fack$, since $t'=t+\fack$ and $z'\leq z-1$.
Combining these two inequalities together with
Part \ref{item:lemma: t<t' => C(t) <= C(t')}
of Lemma \ref{lemma:bmmb: C(t') subset C(t'')},
we get that
$|\cR_{j}(t+z\cdot \fack)| \geq \ell$, as needed.
The lemma follows.
\end{proof}

\subsubsection{The Key Lemma}
\label{subsec: The Key Lemma}

We continue to assume all the context we established earlier in Subsection \ref{subsec:Proof Prelim}.
The lemma below summarizes some helpful complexity bounds.

\begin{lemma}[``Complexity Bounds'']
\label{lemma:bmmb: t d ell values}
\label{lemma: t-inequalities}
The following hold for the $t_{d,\ell}$ values:
  \begin{enumerate}
  \item For $d'\leq d''$, $t_{d',\ell} \leq t_{d'',\ell} $ (monotonically increasing in terms of $d$).
  \label{item:lemma: t dl monotonically increasing}

  \item For $\ell\geq 2$, $d\geq 1$, $t_{d+r,\ell-1} + r\cdot \fack \leq t_{d,\ell}$.
  \label{item:lemma com-bounds r fack}

  \item For $\ell\geq 2$, $d\geq 1$, $t_{d+r,\ell-1} + \fack \leq t_{d-1,\ell}$.
  \label{item:lemma com-bounds fack t d+r ell-1 + Fack leq t d-1 ell}

  \item For $\ell\geq 1$, $d\geq 1$, $t_{d-1,\ell} + \fprog = t_{d,\ell}$.
  \label{item:lemma com-bounds fprog}

  \item For $\ell > 1$, $\ell\cdot \fack \leq t_{0,\ell} $.
  \label{item:t- t0 + ell fack leq t 0 ell}

  \end{enumerate}
\end{lemma}

\begin{proof}
By simple algebraic calculations.
\end{proof}

To prove the key lemma, we show a double induction
for $\ell$ as an ``outer'' induction and for distance $d$ as an ``inner'' induction.
To warm up, let us begin with two special cases.
The first (Lemma \ref{lemma: ell=1} below) will be used in the base case for $\ell=1$ for the outer induction of the inductive proof
in the main lemma.
The second (see Lemma \ref{lemma: d=0}) will be used in the base case for $d=0$ for the inner induction of the inductive proof
in the main lemma.

\begin{lemma}
\label{lemma: ell=1}
Let $j$ be a node at distance $d=d_G(i_0,j)$ from $i_0$.
Then:
\begin{enumerate}
\item
$\cR_j(t_{d,1}) \neq \emptyset$.
\item
$\cC_j(t_{d,1} + \fack) \neq \emptyset$.
\end{enumerate}
\end{lemma}

\begin{proof}
\begin{enumerate}
\item
For Part 1, we use induction on $d$.
For the base case, consider $d=0$.
Then $j = i_0$ and $t_{d,1} = t_{0,1} = 0$.
Since $m \in \cR_{i_0}(0)$, we see that $\cR_j(t_{d,1}) \neq \emptyset$,
as needed.

For the inductive step, assume Part 1 for $d-1$ and prove it
for $d$.
Let $j'$ be the predecessor of $j$ on a shortest path in $G$ from $i_0$ to
$j$; then $d_G(i_0,j') = d-1$.
By inductive hypothesis, we know that
$\cR_{j'}(t_{d-1,1}) \neq \emptyset$.
Then Lemma~\ref{lemma: CG5} implies that $\cR_j(t_{d-1,1} + \fprog) \neq
\emptyset$.
Since $t_{d,1} = t_{d-1,1} + \fprog$, this implies that
$\cR_j(t_{d,1}) \neq \emptyset$, as needed.

\item
Part 2 follows from Part 1 using Lemma~\ref{lemma: CG3}, Part 3,
applied with $k=1$.
\end{enumerate}
\end{proof}

\begin{lemma}
\label{lemma: d=0}
Let $\ell \geq 1$.
Then:
\begin{enumerate}
\item
$m \in \cR_{i_0}(t_{0,\ell})$.
\item
Either $m \in \cC_{i_0}(t_{0,\ell} + \fack)$ or $|\cC_{i_0}(t_{0,\ell} + \fack)| \geq \ell$.
\end{enumerate}
\end{lemma}

\begin{proof}
Since $m \in \cR_{i_0}(0)$, and $0 \leq t_{0,\ell}$, we have $m \in
\cR_j(t_{0,\ell})$, which yields Part 1.

For Part 2, if $m \in \cC_{i_0}(0)$, then clearly $m \in \cC_{i_0}(t_{0,\ell} + \fack)$, which suffices.
So suppose that $m \notin \cC_{i_0}(0)$.
Then $m \in \cR_{i_0}(0) - \cC_{i_0}(0)$, so $m$ is on $bcastq_{i_0}$
in the final state $s_0$ at time $t=0$.

If in state $s_0$, the position of $m$ on $bcastq_{i_0}$ is $\leq
\ell$, then Lemma~\ref{lemma: CG2}, Part 1, implies that
$m \in \cC_{i_0}(\ell \fack)$.
By Part \ref{item:t- t0 + ell fack leq t 0 ell} of Lemma~\ref{lemma: t-inequalities}, $ \ell \fack \leq t_{0,\ell} + \fack$,
which implies together with
Part \ref{item:lemma: t<t' => C(t) <= C(t')}
of Lemma \ref{lemma:bmmb: C(t') subset C(t'')}
that $m \in \cC_{i_0}(t_{0,\ell} + \fack)$,
which is sufficient to establish the claim.

On the other hand, if in state $s_0$, the position of $m$ on
$bcastq_{i_0}$ is strictly greater than $\ell$, then we apply
Lemma~\ref{lemma: CG2}, Part 2, to conclude that
$|\cC_{i_0}( \ell \fack)| \geq \ell$.
That implies that $|\cC_{i_0}(t_{0,\ell} + \fack)| \geq \ell$, which
again suffices.
\end{proof}

And now, for the main lemma.

\begin{lemma}\label{lemma:bmmblemma-full}
  Let $j$ be a node at distance $d=d_G(i_0,j)$ from $i_0$ in $G$.
  Let $\ell$ be any positive integer. Then:
  \begin{enumerate}
  \item  Either $m\in \cR_j(t_{d,\ell})$ or $|\cR_j(t_{d,\ell})|\geq \ell$.

  \item  Either $m\in \cC_j(t_{d,\ell}+\fack)$ or $|\cC_j(t_{d,\ell}+\fack)|\geq \ell$.

  \end{enumerate}

\label{lem:bmmb:time}
\end{lemma}

\begin{proof}
We prove both parts together by induction on $\ell$.
For the base, $\ell= 1$, both statements follow
immediately from Lemma \ref{lemma: ell=1}.
For the inductive step, let $\ell\geq 2$.
We assume the lemma
statement for $\ell-1$ (and for all $d$) and prove it for $\ell$.
To prove the lemma statement for $\ell$, we use a second, ``inner''induction, on the distance
$d$ from $i_0$ and the destination $j$.
For the base, $d = 0$, both statements
follow from Lemma \ref{lemma: d=0}.

\noindent {\bf Inductive Step: $d \geq 1$.}
For the inductive step, we assume $d \geq 1$.
Assume both parts of the lemma for
(1) $\ell -1$ (as ``outer'' induction hypothesis) for all distances;
and (2) for $\ell$ for distance $d-1$ (as ``inner'' induction hypothesis).
We prove both parts of the lemma for $\ell$ and distance $d$.

By our ``outer'' inductive hypothesis, all processors of the network satisfy the two parts of the lemma for $\ell-1$ and all values of $d$.
In particular, by combining the inductive hypothesis for $\ell-1$  and all values of $d$ with
Part \ref{item:lemma: t dl monotonically increasing} of Lemma \ref{lemma:bmmb: t d ell values} and
Part \ref{item:lemma: t<t' => C(t) <= C(t')} of Lemma \ref{lemma:bmmb: C(t') subset C(t'')},
it follows that for every node $j' \in N^r_G(j)$,
either
\begin{itemize}

\item [(S1)] $m\in \cC_{j'}(t_{d+r,\ell-1}+\fack) \mbox{ or } |\cC_{j'}(t_{d+r,\ell-1}+\fack)|\geq \ell-1.$

\end{itemize}
We use distance $d+r$ for $j'$ because $j'$ is at distance at most $d+r$ from $i_0$ in $G$
($j$ is at distance $d$ from $i_0$ in $G$; and $j'$ is either $j$ itself or it at distance at most $r$ from $j$ in $G$, since $j' \in N^r_G(j)$).
Let $t^*= t_{d+r,\ell-1}+\fack$.
Recall that,
by Part \ref{item:lemma com-bounds r fack} of Lemma \ref{lemma:bmmb: t d ell values},
\begin{equation}
t^* + (r-1) \fack \leq t_{d,\ell}~.
\label{ineq:lemma: t d+r ell-1 leq t d ell}
\end{equation}

\begin{enumerate}

\item We now prove Part 1 of the lemma (for $\ell$ and $d$).
Suppose that $m\in \cC_j(t^*)$  or $|\cC_{j}(t^*)| \geq \ell$.
Then, either
$m\in \cR_{j}(t_{d,\ell})$ or $|\cR_{j}(t_{d,\ell})| \geq \ell $,
since $\cC_{j}(t^*) \subseteq \cR_{j}(t_{d,\ell})$,
by Inequality (\ref{ineq:lemma: t d+r ell-1 leq t d ell}) and Part \ref{item:lemma: C_i(t) subseteq R_i(t)} of Lemma \ref{lemma:bmmb: C(t') subset C(t'')}
(with $t=t^*$ and $t'= t_{d,\ell}$).
This implies that $j$ satisfies Part 1 of the lemma statement for $\ell$. This implies Part 1.
Now, suppose the contrary, that $m\not\in \cC_j(t^*)$ and $|\cC_{j}(t^*)| < \ell$.
Since, $j$ does satisfy (S1) for $\ell-1$, it follows that $|\cC_{j}(t^*)| \geq \ell-1$,  since $m\not\in \cC_j(t^*)$.
Thus, in the remaining case, we have
\begin{equation}
m\not\in \cC_j(t^*) \mbox{ and } |\cC_{j}(t^*)| = \ell-1.
\label{ineq:lemma: m not in Cj(t*) and |cj(t*)|= ell-1}
\end{equation}
We next prove that $|\cR_{j}(t_{d,\ell})| \geq \ell$ (which implies Part 1 of the lemma).
We consider two cases regarding the set of messages that $j$ completes by time $t^*$
and the sets of messages that are completed (by that time) by all neighbors of $j$ in $G^r$.

\noindent {\bf Case 1:} There exists some neighbor $j'$ of $j$ in $G^r$ such that $\cC_{j'}(t^*) \not =  \cC_j(t^*)$.

Choose a closest node $j''$ to $j$ in $G$ with this property, and let $j^*$
be the next-closer node on some shortest path from $j''$ to $j$ in $G$.
That is,
$j'' \in \arg\min\{d_G(j,i') \mid \cC_{i'}(t^*)\not= \cC_{j}(t^*) \}$, $(j'',j^*)\in E$ and $d_G(j,j^*)=d_G(j,j'')-1$.
Recall that, in this case, there exists such a neighbor $j''\in N^r_G(j)$ with this property, hence $0\leq d_G(j,j^*)<d_G(j,j'')\leq r$.

To apply Lemma \ref{lemma: progress in z cdot Fack}, with $z = d_G(j,j^*)\leq r-1$ and $t = t^*$, we first need to show that the three hypothesis of the lemma hold.
First, by the second conjunct of (\ref{ineq:lemma: m not in Cj(t*) and |cj(t*)|= ell-1}),
we have $|\cC_{j}(t^*)| = \ell-1$, which implies the first hypothesis of Lemma
\ref{lemma: progress in z cdot Fack}.
Second,
we need to show that $\cC_{j}(t^*) \subseteq \cC_{j'}(t^*)$, for every $j' \in N_G^z(j)$.
This follows from the fact that $d_G(j,j'')=z+1$ and the fact that $j''$ is a closest node to $j$ in $G$ with the property that $\cC_{j''}(t^*) \neq \cC_{j}(t^*)$.
This implies that $\cC_{j}(t^*) = \cC_{j'}(t^*)$, and in particular $\cC_{j}(t^*) \subseteq \cC_{j'}(t^*)$, for every $j' \in N_G^z(j)$, as needed.
Third, we need to show that $|\cR_{j'}(t^*)|\geq \ell$ for some neighbor $j' \in N^z_G(j)$.
We show that $|\cR_{j^*}(t^*)|\geq \ell$ (that is, $j^*$, in particular, does satisfy this property).

The fact that $j''$ is a closest node with this property and node $j^*$ is closer than $j''$ to $j$ in $G$, implies that
$\cC_{j^*}(t^*) = \cC_j(t^*)$ and that $\cC_{j''}(t^*) \neq \cC_{j^*}(t^*)$.
By the inductive hypothesis for $\ell-1$, we obtain that either $m\in \cC_{j''}(t^*)$ or $|\cC_{j''}(t^*)| \geq \ell-1$;
either way, by Inequality
(\ref{ineq:lemma: m not in Cj(t*) and |cj(t*)|= ell-1}),
there is some message $m'\in \cC_{j''}(t^*) \setminus \cC_{j^*}(t^*)$,
which implies by
Part \ref{item:lemma: i and j neigh t'<=t C_i(t') subseteq R_j(t)}
of Lemma \ref{lemma:bmmb: C(t') subset C(t'')}
(with $t=t'=t^*$, $i'=j''$ and $i=j^*$),
that $m'\in \cR_{j''}(t^*) \setminus \cC_{j''}(t^*)$.
This, in turn implies that
$|\cR_{j^*}(t^*)| \geq \ell$.
Then Lemma \ref{lemma: progress in z cdot Fack}
(using $z = d_G(j,j^*)\leq r-1$ and $t = t^*$),
yields that $|\cR_j(t^* + (r-1) \fack)| \geq \ell$.
Since $t^*= t_{d+r,\ell-1}+ \fack$,
by Part \ref{item:lemma com-bounds r fack}
of Lemma \ref{lemma:bmmb: t d ell values},
we have that $t_{d,\ell} \geq t^* + (r-1) \fack$.
Thus, $|\cR_j(t_{d,\ell})| \geq \ell$ as needed for Part 1 of the lemma.

\noindent {\bf Case 2:} $\cC_{j'}(t^*) = \cC_j(t^*)$, for all neighbors $j'\in N^r_G(j)$.

Since $G' \subseteq G^r$
\footnote{
Note that this is the first place that we use this assumption.
}, it holds, in particular, that $\cC_{j'}(t^*) = \cC_{j}(t^*)$, for all neighbors $j'$ of $j$ in $G'$.
Let's focus on time $t_{d-1,\ell}$~.
By Part \ref{item:lemma com-bounds fack t d+r ell-1 + Fack leq t d-1 ell}
and Part \ref{item:lemma com-bounds fprog}
of Lemma \ref{lemma:bmmb: t d ell values},
we have
\begin{equation}
t^* \leq t_{d-1,\ell} + \fprog = t_{d,\ell} ~.
\label{ineq:lemma: t^* leq t d-1 ell + fprog = t d ell}
\end{equation}
Now, if $|\cR_{j}(t_{d-1,\ell})| \geq \ell$, then $|\cR_{j}(t_{d,\ell})| \geq \ell$, by
Part \ref{item:lemma: t dl monotonically increasing}
of Lemma \ref{lem:bmmb:time} and we are done.
So suppose that $|\cR_{j}(t_{d-1,\ell})| \leq \ell-1$.
Recall that
\begin{equation}
|\cC_{j}(t^*)|=\ell-1 \mbox{ and } \cC_{j}(t^*) \subseteq \cR_{j}(t_{d-1,\ell}),
\label{eq-subset:lemma: Cj(t*) with t d-1 ell}
\end{equation}
where the second inequality holds by combining
Inequality (\ref{ineq:lemma: t^* leq t d-1 ell + fprog = t d ell})
together with Part \ref{item:lemma: C_i(t) subseteq R_i(t)}
of Lemma \ref{lemma:bmmb: C(t') subset C(t'')}
(with $t=t^*$ and $t'=t_{d-1,\ell}$).
This implies that $|\cR_{j}(t_{d-1,\ell})| \geq \ell-1$.
Hence, $|\cR_{j}(t_{d-1,\ell})| = \ell-1$, which implies together with Inequality
(\ref{eq-subset:lemma: Cj(t*) with t d-1 ell})  that
\begin{equation}
\cR_{j}(t_{d-1,\ell})=\cC_{j}(t^*).
\label{eq:lemma: Rj t d-1 ell = Cj(t*)}
\end{equation}

Now we will apply Lemma \ref{lemma: z-neighborhood prog bound}, with $t=t_{d-1,\ell}$~.
To do this, we need to show the two hypotheses of that lemma:
First, we
need to show that $\cR_j(t_{d-1,\ell}) \subseteq  \cC_{j'}(t_{d-1,\ell})$ for every neighbor $j'$ of $j$ in $G'$.
Consider some neighbor $j'$ of $j$ in $G^r$.
We have
$$\cR_{j}(t_{d-1,\ell}) = \cC_{j'}(t^*) \subseteq \cC_{j'}(t_{d-1,\ell}), \mbox{ for every neighbor } j' \mbox{ of } j \mbox{ in } G^r,$$
where the first equality holds by the case analysis assumption and Equality (\ref{eq:lemma: Rj t d-1 ell = Cj(t*)});
and the second inequality holds by combining the first inequality of
(\ref{ineq:lemma: t^* leq t d-1 ell + fprog = t d ell})
with
Part \ref{item:lemma: t<t' => C(t) <= C(t')}
of Lemma \ref{lemma:bmmb: C(t') subset C(t'')}.
This implies, in particular, that
$\cR_{j}(t_{d-1,\ell}) \subseteq \cC_{j'}(t_{d-1,\ell})$, for every neighbor $j'$ of $j$ in $G'$ (since $G' \subseteq G^r$), as needed for the first hypothesis of Lemma \ref{lemma: z-neighborhood prog bound}.

To show the second hypothesis of Lemma \ref{lemma: z-neighborhood prog bound},
we need to show that
$\cR_{j'}(t_{d-1,\ell}) - \cR_{j}(t_{d-1,\ell})\not= \emptyset$, for some neighbor $j'$ of $j$ in $G$.
So, fix a neighbor $j^*$ of $j$ in $G$ at distance $d-1$ from $i_0$.
By the inductive hypothesis for $d$,
we obtain that either $m\in \cR_{j^*}(t_{d-1,\ell})$ or $|\cR_{j^*}(t_{d-1,\ell})| \geq \ell$;
either way,
by Inequality
(\ref{ineq:lemma: m not in Cj(t*) and |cj(t*)|= ell-1}),
there is some message $m'\in \cR_{j^*}(t_{d-1,\ell}) \setminus \cC_{j}(t^*)$.

Then, Lemma \ref{lemma: z-neighborhood prog bound}
yields that $|\cR_{j}(t_{d-1,\ell}+ \fprog)| > |\cR_{j}(t_{d-1,\ell})| = \ell-1$.
This implies that $|\cR_{j}(t_{d,\ell})| \geq \ell$, since $t_{d,\ell} = t_{d-1,\ell}+\fprog$
(by Part \ref{item:lemma com-bounds fprog}
of Lemma \ref{lemma:bmmb: t d ell values}).
Part 1 of the lemma follows.

\item Now, we prove Part 2 of the lemma.
Before proceeding recall that
$t^*=t_{d+r,\ell-1}+\fack \leq t_{d,\ell} < t_{d,\ell}+\fack$
(where the left inequality holds
by Part
\ref{item:lemma com-bounds fack t d+r ell-1 + Fack leq t d-1 ell}
of Lemma
\ref{lemma:bmmb: t d ell values}), which implies together with
Part \ref{item:lemma: t<t' => C(t) <= C(t')}
of Lemma \ref{lemma:bmmb: C(t') subset C(t'')}, that
$$
\cC_j(t^*) \subseteq \cC_j(t_{d,\ell}) \subseteq \cC_j(t_{d,\ell}+\fack).
$$
Now, suppose that $m\in \cC_j(t^*)$.
Then,
$m\in \cC_{j}(t_{d,\ell}+\fack)$
(since $\cC_{j}(t^*) \subseteq \cC_{j}(t_{d,\ell}+\fack)$) and we are done.
Next, assume that $m\not\in \cC_j(t^*)$.
Then,  by the inductive hypothesis for $\ell-1$, $|\cC_j(t^*)| \geq \ell-1$.
This implies, in particular, that
$|\cC_j(t_{d,\ell})| \geq \ell-1$
(since $\cC_j(t^*)\subseteq \cC_j(t_{d,\ell})$).
By Part 1, either $m\in \cR_j(t_{d,\ell})$ or $|\cR_j(t_{d,\ell})|\geq \ell$;
either way, we obtain that $|\cR_j(t_{d,\ell})|\geq \ell$, since, $\cC_j(t^*)\subseteq \cR_j(t_{d,\ell})$, $|\cC_j(t^*)|\geq \ell-1$ and $m\not\in \cC_j(t^*)$.
Then Corollary \ref{coro:bmmb: ack and prog progress} implies that
$|\cC_j(t_{d,\ell}+\fack)|>|\cC_j(t_{d,\ell})|$, so
$|\cC_j(t_{d,\ell}+\fack)|\geq \ell$ as needed.

\end{enumerate}
\end{proof}

\subsubsection{The Main Theorem}

Let $\mathcal{K}\subseteq \cM$ be the set of messages that arrive at the nodes in a given execution $\alpha$.

\begin{theorem}
\label{thm:bmmb:time-full}
If $|\mathcal{K}| \leq k$ then $\cR_j(t_1)=\mathcal{K}$ for every node $j$,
where $t_1 = (D + (r+1)k - 2) \fprog + r(k-1) \fack$.
\end{theorem}
%
%
%
%
%
%
%
%
%
%
The conclusion of this theorem says all the messages of $\mathcal{K}$ are received at all
nodes by time $t_1$.
\begin{proof}
Follows directly from Lemma \ref{lemma:bmmblemma-full}.
%
%
%
%
%
\end{proof}

}
\subsection{Lower Bound for Grey Zone G'}
\label{sec:lower}

We are left with two questions concerning MMB. First, does BMMB perform as well
as the $r$-restricted case when we we consider other natural restrictions on $G'$, such as the grey zone constraint?
And second, if not BMMB, are there {\em any} MMB algorithms for the standard abstract MAC layer model
model that can perform well given the grey zone restriction, or perhaps even perform well for arbitrary $G'$?
Below, we answer both questions in the negative by proving that all MMB algorithms
require  $\Omega((D+k)\fack)$ time to solve MMB with a grey zone restricted $G'$.
This result establishes that our analysis of BMMB from Section~\ref{sec:bmmb:arbitrary} is tight,
and it opens an intriguing gap between
the superficially similar $r$-restricted and grey zone constraints.

\begin{theorem}
For any Multi-Message Broadcast algorithm ${\cal A}$, every $k>1$, and random bit assignment,
there exists a network, message assignment, and message scheduler such
that ${\cal A}$ requires $\Omega((D+k)\fack)$ time to solve the MMB problem.
\label{thm:lower}
\end{theorem}

To prove our main theorem, we handle the $k\fack$ and $D\fack$ terms separately. The $k\fack$ part is simple:
consider a node $u$ that represents the only bridge in $G$ between a receiver $v$ and the source(s) of $k$ messages.
Our message size limit restricts $u$ to send only a constant number of messages to $v$ at a time,
inducing the $\Omega(k\fack)$ bound.

\begin{lemma}
For every $k\in[1,n-2]$, algorithm ${\cal A}$, and random bit assignment,
there exists a network with $G'=G$, 
a message assignment that has no node begins with more than one message---what we call a {\em singleton assignment}---and a message scheduler such that ${\cal A}$
requires $\Omega(k\fack)$ time to solve the MMB problem.
\label{lem:lower:1}
\end{lemma}
\begin{proof}
We first fix our definition of $G'=G$.
To define $G$,
connect each node in $U = \{u_1,u_2,...,u_{k-1}\}$ to $u_k$, forming a star.
Then connect $u_k$ to some other node $v$.
Start each node in $U \cup \{u_k\}$ with a unique broadcast
message.
Consider a message schedule
that requires the full $\fack$ time between each broadcast and its corresponding acknowledgment.
We now bound the time for $v$ to receive all $k$ messages.
The key observation is that $u_k$ is a {choke-point} through which all messages
must pass to arrive at $v$.
To bound the time for messages to make it through this constriction,
 divide time into {\em rounds} of length $\fack$.
By our assumption on the scheduler, $u_k$ can begin the transmission of at most a constant number of message per round.
Therefore, $v$ can receive at most a constant number of new messages per round.
The $\Omega(k\fack)$ bound follows directly.
\end{proof}

The more interesting step is proving the $\Omega(D\fack)$ term. 
  %
  %
  %
 %
%
To accomplish this goal,
 we begin by defining the network used in our proof argument.
 Fix some diameter $D$ that is divisible by $2$ (the below proof can be easily modified to handle odd
 $D$).
 Fix two node sets  $U_A = \{a_1,a_2,...,a_D\}$ and $U_B = \{b_1,b_2,...,b_D\}$.
  Let $A$ and $B$ be the two line graphs that connect, in order of their indices, the nodes in $U_A$ and $U_B$, respectively.
 Let $C$ be the dual graph network over nodes $U_A \cup U_B$
 where: $G$ consists of the edges in $A$, $B$,
 and $G'$ is defined to include all the edges in $G$,
 as well as the following extra edges:
 for $i<D$, $a_i$ (resp. $b_i$) is connected to $b_{i+1}$ (resp. $a_{i+1}$).
Notice that our definition of $G'$ in $C$
satisfies the definition of grey zone restricted for a sufficiently large value for
the constant $c$ (see Section~\ref{sec:model}). An example of this network is shown in \Cref{fig:LB}.

\begin{figure*}[t]
	\centering
		\includegraphics[width=0.80\textwidth]{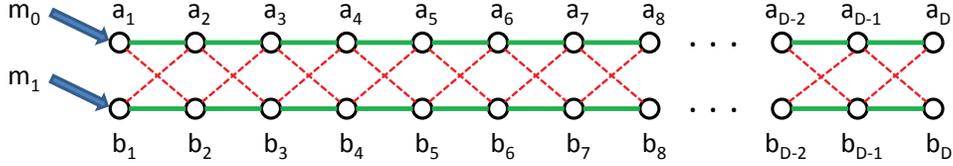}
	\caption{{\small The lower bound network. The green solid lines represent the reliable edges, i.e., those in $G$, and the red dashed lines represent the unreliable edges, i.e., those in $G'\setminus G$.}}
	\label{fig:LB}
\end{figure*}

 In the following, we define a {\em endpoint-oriented} execution
 to be an execution of an MMB broadcast algorithm for $k=2$, in $C$, for message set $M=\{m_0, m_1\}$,
 where $m_0$ starts at $a_1$ and $m_1$ starts at $b_1$.
 Given a finite endpoint-oriented execution $\alpha$,
 let $\ell_0(\alpha)$ be the largest node in $A$ (by increasing index order) that has received $m_0$ in $\alpha$,
 and $\ell_1(\alpha)$ be the largest node in $B$ that has received $m_1$.
 Let $q_0$ and $q_1$ be defined in the same way, except now capturing the largest node in the relevant line
  to have received {\em and initiated a broadcast}
 of the relevant message.
 Finally, we call an execution (or execution prefix/extension) {\em valid}, if the message
 events and their timing satisfy the model guarantees.

 The main insight in our proof argument is that given the right scheduling strategy,
  $m_0$'s progress down $A$ can slow $m_1$'s progress down $B$,
 as well as the other way around.
 This strategy, however, requires that nodes that receive these messages proceed to then broadcast them as well.
 The below lemma argues that when $m_0$ and $m_1$ make progress, either
 the next hops start broadcasting, or we can show that at least one of the messages
 is delayed long enough to establish our desired result.

\begin{lemma}
Let $\alpha$ be a finite endpoint-oriented execution of a MMB algorithm ${\cal A}$  with
 random bit assignment $\kappa$ in network $C$, such that
$q_0(\alpha) < \ell_0(\alpha)=a_i$ and $q_1(\alpha) < \ell_1(\alpha)=b_i$ for some $i\in \{1,...,D-1\}$.
There exists a message schedule that produces one of the following outcomes:
(1) an extension $\alpha'$ of $\alpha$ in which no time passes and
$q_0(\alpha') = \ell_0(\alpha')=a_{i}$ and $q_1(\alpha') = \ell_1(\alpha')=b_{i}$; or
(2) an extension $\alpha''$ of $\alpha$ of duration $\Omega(D\fack)$ in which the MMB problem is not yet solved.
\label{lem:lower:2}
\end{lemma}
\begin{proof}
Let $\alpha$ be the finite execution specified by the lemma statement.
By assumption $a_i$ has $m_0$ and $b_i$ has $m_1$ at the end of $\alpha$,
but neither node has yet initiated a broadcast of this message.
We begin by extending $\alpha$ using the following message schedule behavior:

{\em For every broadcast initiated by a node in $U_A \cup U_B \setminus \{a_i,b_i\}$,
or broadcast by $a_i$ (resp. $b_i$) but not containing $m_0$ (resp. $m_1$)}:
 deliver the message to the broadcaster's neighbors in $G$ (but to {\em no} $G'$-only neighbors)
 and then return an acknowledgment to the broadcaster, instantaneously (i.e., with no time passing).
In scheduling these events, construct the schedule to proceed in a round robin fashion through all nodes;
that is, for each node in this order, if there is a receive or acknowledgment event to schedule (as specified
in the above rule), schedule that event and allow the node to initiate its next broadcast (if its algorithm dictates),
then move on to the next node in the round robin order.

Call this extension $\beta$. Notice, $\beta$ is not necessarily a valid execution of our algorithm because
if $a_i$ or $b_i$ initiate a broadcast of $m_0$ and $m_1$, respectively, in $\beta$, they are starved by the schedule.
We now use $\beta$, however, to force our algorithm to satisfy one of the two lemma outcomes.
In more detail, let $s_a$ be the step in $\beta$ where $a_i$ initiates a broadcast containing $m_0$ (define $s_a=\bot$
if no such step exists). Define $s_b$ the same with respect to $b_i$ and $m_1$.
Because our schedule in $\beta$ never allows these broadcasts to complete, there can only be at most one such step for
$a_i$ and $b_i$ in $\beta$.
We consider three cases depending on the values of $s_a$ and $s_b$.

{\em Case $1$: Assume that $s_a \neq \bot$ and $s_b\neq \bot$.}
Let $\alpha'$ be the prefix of $\beta$ that stops at whichever of these two steps happens later in $\beta$.
Notice, $\alpha'$ provides a valid extension of $\alpha$: even though either $a_i$ or $b_i$ might have been
delayed from delivering a message in this extension, no time passed between $\alpha$ and the end of $\alpha'$,
so no timing guarantees were violated.
Accordingly, we see that $\alpha'$ satisfies outcome (1) of the lemma statement.

{\em Case $2$: $s_a = s_b = \bot$.}
In this case, $\beta$ does not starve any node: every initiated broadcast is delivered to $G$ neighbors and acknowledged.
Let $\alpha'$ be a transformation of $\beta$ where we: (1) allow $\fprog$ time to pass between each broadcast and its corresponding acknowledgment;
and (2) we stop after $D\fack$ time has passed since the end of $\alpha$.
 Because our algorithms are event-driven (and therefore have no concept of time),
 it is straightforward to see that $\alpha'$ is indistinguishable from $\beta$ for all nodes.
 We also node that the schedule in $\alpha'$ satisfies the necessary time constraints,
 as we never delay a pending delivery by more than $\fprog$ time.
It follows that $\alpha$ satisfies outcome (2) of the lemma statement.

{\em Case $3$: either $s_a=\bot$ or $s_b=\bot$, but not both.}
Assume, w.l.o.g., that $s_a=\bot$ (the other case is symmetric).
Let $\beta'$ be an extension of $\alpha$ defined with the same rules as $\beta$ with two exceptions:
(1) schedule $b_i$'s broadcasts of $m_1$ the same as all other broadcasts; and
(2) allow $\fprog$ time to pass between each broadcast and its corresponding acknowledgment.
Let $\alpha'$ be the prefix of $\beta'$ that ends after $D\fack$ time has passed since the end of $\alpha$.
As in the previous case, we note that $\alpha'$ is indistinguishable from $\beta$ with respect to
nodes in $A$ (the scheduling rules defined above for $\beta$ do not allow messages
from $B$ to be delivered to nodes in $A$, therefore nodes in $A$ cannot learn that, in $\beta'$,
$b_i$ can succeed in its broadcasts of $m_1$) and that it still satisfies the model's time bounds.
As a result, $a_i$ behaves the same in $\alpha'$ as in $\beta$ and does not broadcast $m_0$. Because
$a_i$, by assumption, is the furthest node down the line in $A$ to receive $m_0$ so far,
it follows that by the end of $\alpha'$ there are nodes in $A$ that have not yet received
$m_0$.
It follows that $\alpha'$ satisfies outcome (2).
%
%
%
%
\end{proof}

\noindent With the above lemma established, we
can use it to prove the main lemma regarding the necessity of $D\fack$ rounds.
Here is the main idea: As the messages arrive at each new hop in their respective lines,
we apply the above lemma to force these new hops to initiate broadcasts (or, directly prove our time bound by delaying too long).
Once we have established that the message frontiers on each line are broadcasting,
we can allow these broadcasts to mutually interfere over $G' \setminus G$ edges in such a way that satisfies
the progress bound while preventing useful dissemination. 

\begin{lemma}
For every algorithm ${\cal A}$ and random bit assignment $\kappa$,
there exists a message assignment and schedule such that ${\cal A}$ requires $\Omega(D\fack)$ time
to solve MMB in network $C$ for $k=2$.
\label{lem:lower:3}
\end{lemma}
\begin{proof}
We construct an endpoint oriented execution of ${\cal A}$ in $C$ with random bits $\kappa$,
by defining the message schedule behavior.
We start with $\alpha_0$: the finite execution that captures the behavior of the above
system only through $a_1$ receiving $m_0$ and $b_1$ receiving $m_1$.
These events happen at the beginning of the execution, so no time passes in $\alpha_0$.

Notice, $\alpha_0$ satisfies the preconditions required to apply Lemma~\ref{lem:lower:2}.
Apply this lemma to $\alpha_0$.
By the definition of this lemma, there are two possible outcomes.
If it is the second outcome, we have proved our theorem.
Assume, therefore, that the lemma produces an extension $\alpha_0'$ of $\alpha_0$ that
satisfies  the first outcome.
At the end of $\alpha_0'$,
we know that $a_1$ has initiated a broadcast of $m_0$ and $b_1$ has
initiated a broadcast of $m_1$,
and no time has passed since these broadcasts are initiated.
We further note that at this point, $m_0$ has made it no further down the $A$ line and $b_1$ has
made it no further down the $B$ line.

We now extend $\alpha_0'$ with a message schedule that delays $m_0$'s arrival
at $a_2$ and $m_1$'s arrival at $b_2$ by the maximum $\fack$ time.
To do so, partition an interval of $\fack$ time following $\alpha_0'$ into sub-intervals of length $\fprog$.
At the end of each sub-interval,
deliver $m_0$ from $a_1$ to $b_2$ (over a $G'$ edge)
and $m_1$ from $b_1$ to $a_2$ (also over a $G'$ edge).
At the end of this $\fack$ interval,
allow $m_0$ to make it to $a_2$ and $m_1$ to make it to $b_2$,
and acknowledge these broadcasts.
Notice, this schedule satisfies both the progress and acknowledgment
bounds for $a_1$ and $b_1$'s broadcasts during this interval.
During this $\fack$ interval, however, we must also schedule {\em other nodes'} broadcasts.
To do so, we use a simple rule: for every other broadcast,
allow the message to be delivered to all (and only) $G$ neighbors and be acknowledged
at the end of the next $\fprog$ interval.

Notice, our above delay strategy leads to a finite execution $\alpha_1$,
of duration $\fack$ longer than $\alpha_0$,
where $q_0(\alpha_1) \neq \ell_0(\alpha_1)=a_2$ and $q_1(\alpha_1) \neq \ell_1(\alpha_1)=b_2$.
We can, therefore, apply our above argument again, now replacing $\alpha_0$ with $\alpha_1$.
Indeed, we can keep applying this argument until either
we arrive at outcome (2) from Lemma~\ref{lem:lower:2},
or we build up to $\alpha_{D-2}$,
an execution of length $\Omega(D\fack)$ in
which $m_0$ and $m_1$ have not yet made it to the end of the $A$ and $B$ lines, respectively.
Either way, we have proved the theorem statement.
\end{proof}


\section{Multi-Message Broadcast with an Enhanced Abstract MAC Layer}\label{sec:FMMB}

In Section~\ref{sec:standard},
we proved that in the standard abstract MAC layer model, $\Omega(k\fack)$ time is necessary to solve MMB,
and for some definitions of $G'$, an additional
$\Omega(D\fack)$ time is also necessary.
Our analysis of BMMB then established that this algorithm is essentially the best you can do in this model.
In this section, we tackle the question of how much {\em additional power} we must add to our
model definition to enable faster solutions under the assumption that $\fprog \ll \fack$,
pointing to the extra assumptions of the enhanced abstract MAC layer model as one possible answer.
In particular,
we describe a new algorithm, which we call {\em Fast Multi-Message Broadcast} (FMMB),
that guarantees the following time complexity when run in the enhanced abstract MAC layer model with a grey zone restricted $G'$:

\begin{theorem}\label{thm:FMMB}
The FMMB algorithm solves the MMB problem in $O((D\log n + k\log{n} + \log^3{n})\fprog)$ time, w.h.p.,
in the enhanced abstract MAC layer model and grey zone restricted $G'$.
\label{thm:fmmb}
\end{theorem}

\noindent This result has no $\fack$ term. As the size of $\fprog$ decreases, this result's advantage over BMMB increases.


\bigskip 

\noindent {\bf Preliminaries.}
In the following, for $v\in V$, we use $ID(v)$ to refer to $v$'s unique id, $N_G(v)$ to describe {\em the ids} of $v$'s neighbors in $G$,
and $N_{G'}(v)$ to describe the ids of $v$'s neighbors in $G'$. We use ${\cal M}$ to refer to the set of messages to be disseminated in a given execution of MMB.
We call a set $S \subseteq V$ of nodes \emph{$G$-independent} if for each pair of nodes $u, v \in S$, we have $(v,u)\notin E$.
\fullOnly{In our analysis, we make frequent use of the following well-known fact, sometimes referred to as the \emph{Sphere Packing} Lemma.\footnote{Although the precise constants in the bound on the cardinality of $S$ in \Cref{lem:Sphere-Packing} are known, the asymptotic version stated above is sufficient for our purposes.}

\begin{lemma}\label{lem:Sphere-Packing}
Consider $P \subseteq \mathbb{R}^2$ such that $\forall p_1\neq p_2 \in P$, we have $1<\left\|p_1 - p_2\right\|_2\leq d$. Then $|P|=O(d^2)$.
\end{lemma}}

\subsection{Algorithm Outline}
The FMMB algorithm divides time into lock-step {\em rounds} each of length
$\fprog$. This can be achieved by leveraging the ability of a node to use time and {\em abort} a broadcast in progress in 
the enhanced abstract MAC layer. In more detail, when we say a node {\em broadcasts in round $t$}, we mean that it initiates the broadcast at the beginning
of the time slot dedicated to round $t$, and aborts it (if not completed yet) at the end of the time slot.

The FMMB algorithm uses three key subroutines which we summarize here,
but detail and analyze in the subsections that follow. All three subroutines are randomized and will be
shown to hold with sufficiently
high probability that their correctness guarantees can be combined with a union bound.
The FMMB algorithms begins by having nodes construct a maximal independent set (MIS) in $G$ using $O(\log^3{n})$ rounds.
We note that this MIS subroutine might be of independent interest.\footnote{The previously best known MIS solution for an abstract MAC layer model uses time that is linear in $n$~\cite{Lynch:2012}.}
The FMMB algorithm then uses a gather subroutine to gather the broadcast messages at nearby MIS nodes in an additional $O(k+\log{n})$ rounds.
Finally, it uses an overlay dissemination subroutine that broadcasts the messages to all MIS nodes, and then to their neighbors
(i.e., all nodes), in $O((D+k)\log{n})$ rounds.
The total combined running time of FMMB is therefore $O( D\log n + k\log n + \log^3 n)$ rounds, which
requires $O((D\log n + k\log{n} + \log^3{n})\fprog)$ total time.

\fullOnly{We continue by explaining each of the three subroutines. Theorem~\ref{thm:fmmb} follows directly Lemmata~\ref{lem:fmmb:mis}, \ref{lem:fmmb:gather}, and \ref{lem:fmmb:spread}. We also note that all three subroutines depend on the assumption of a grey zone restricted $G'$, which is leveraged in our analysis to enforce useful regionalization properties on the MIS nodes.}

\shortOnly{We continue by describing each of the three subroutines. Theorem~\ref{thm:fmmb} follows directly from the analysis of these subroutines.
Due to space constraints, we defer this analysis to the full version of this paper~\cite{full-version}.}

\subsection{The MIS Subroutine}\label{subsec:MIS}
We now describe an MIS subroutine that succeeds in building an MIS in $G$ in $O(c^4\log^3 n)$ rounds, w.h.p., where $c$ is the universal
constant from the grey zone definition (see Section~\ref{sec:model}). In more detail, the algorithm runs for a fixed length of time, $t_{MIS} \in O(c^4\log^3 n)$.
At the end of this period, some set $S\subseteq V$ of nodes {\em join} the MIS.
The algorithm guarantees, w.h.p.,\footnote{For this guarantee, as with the other subroutines we consider, we assume
that the high probability is of the form $1-n^{-c}$ for a sufficiently large constant $c$ to enable a union bound on the fail probability
for all three subroutines.}  that $S$ is a {\em maximal $G$-independent set}: (1) all pairs of nodes in $S$ are $G$-independent; and (2) every $u\in V$ is either
in $S$ or neighbors a node in $S$ in $G$.

 
The subroutine (called ``algorithm" from here forward)
works as follows: initially, all nodes are active. In the course of the algorithm, some nodes join the MIS and some nodes become inactive. The algorithm runs in $O(c^2\log^2 n)$ phases, each of which consists of $O(c^2\log n)$ rounds, which are divided into two parts: \emph{election} and  \emph{announcement}.

The election part has $4\log n$ rounds. At the start, each active node $v$ picks a random bit-string $b(v) \in \{0,1\}^{4\log n}$. In each round $\tau\in [1, 4\log n]$ of this part, each active node $v$ broadcasts its bit-string $b(v)$ iff the $\tau^{th}$ bit of $b(v)$ is $1$. If node $v$ did not broadcast but it received a message $b(u)$, be it from a $G$ or a $G'$ neighbor, then node $v$ becomes temporarily inactive for the rest of this phase. At the end of $4\log n$ rounds of the election part, if a node $v$ is still active, then $v$ joins the MIS set $S$.

The announcement part has $O(c^2 \log n)$ rounds. In each round, each node $v$ that joined the MIS in this phase broadcasts a message containing $ID(v)$ with probability $\Theta(1/c^2)$, and does not broadcast any message with probability $1-\Theta(1/c^2)$. If a node $u$ that has not joined the MIS receives a message $ID(v)$ from a $G$-neighbor, then $u$ knows that one of its $G$-neighbors is in the MIS and thus node $u$ becomes permanently inactive. At the end of the announcement part, each node that joined the MIS in this phase becomes permanently inactive, while each temporarily inactive node becomes active again.

\fullOnly{

\smallskip


\begin{lemma}\label{lem:independence}
The set $S$ of nodes that join the MIS nodes is $G$-independent {with high probability.}
\end{lemma}
\begin{proof}
We show that the probability that there are two $G$-neighbors $v, u \in S$ is at most $\frac{1}{n}$.

First suppose that there are two $G$-neighbors $v$ and $u$ that joined the MIS in the same phase. Then, it must hold that in that phase $b(v)=b(u)$. This is because, otherwise there would exist a round of the election part where one of the two nodes, say $v$, is not broadcasting but the other one, $u$, is broadcasting. In that case, $v$ would receive the message of a $G'$-neighbor---which might be $v$ or not---and thus become temporarily inactive which means that $v$ would not join the MIS in this phase. It is an easy calculation to see that the probability that $b(v)=b(u)$ is at most $\frac{1}{n^4}$ and a union bound over all choices of the pair $u,v$ establishes that the probability of existence of such a pair is at most $\frac{1}{n^2}$.

Now suppose that there were not two $G$-neighbors that joined the MIS in the same phase but  there were two nodes that joined the MIS in different phases. Let $t$ be the first phase in which there are two $G$-neighbors $v$, $u$ that are in the MIS. Without loss of generality, suppose that $v$ was not in the MIS at the end of phase $t$ and $u$ joined the MIS in phase $t' < t$. The set $S'$ of $G'$-neighbors of $v$ that joined the MIS in phase $t'$ is a $G$-independent set. Hence, using \Cref{lem:Sphere-Packing}, we get that $|S'|=O(c^2)$. Now in each round of the announcement part of $t'$, node $u$ broadcasts with probability $\Theta(1/c^2)$ and each other node in $S'$ does not broadcast with probability $\Theta(1/c^2)$. Hence, the probability that $v$ receives the message of $u$ in one round is at least $\Theta(1/c^2) (1-\Theta(1/c^2))^{|S'|-1} \geq \Theta(1/c^2) (1-\Theta(1/c^2))^{O(c^2)} = \Theta(1/c^2)$. Hence, during the $\Theta(c^2 \log n)$ rounds of the announcement part, $v$ receives the message of $u$ with probability at least $1-\frac{1}{n^4}$. Hence, the probability does not receive this message and later joins the MIS is at most $\frac{1}{n^4}$. Again, taking a union over all node pairs shows that the probability of existence of such a pair $u$, $v$ is at most $\frac{1}{n^2}$.

Overall, using another union bound over the two cases considered in the above two paragraphs, we get that the probability that there are two $G$-neighbors $v,u \in S$ is at most $\frac{1}{n}$.
\end{proof}


\begin{lemma} \label{lem:MISgrowth} For each phase $t$ and each node $v$. If at the start of phase $t$, node $v$ is active, then in phase $t$, at least one node $u$ such that $\left\|p(v)-p(u)\right\|_2 =O(c\log n)$ joins the MIS.
\end{lemma}
\begin{proof}
Consider a phase $t$ and a node $v$ that is active at the start of phase $t$. We use a node variable $u_{\tau}$, for round $\tau\in [1, 4\log n]$ of the election part, to keep track of a node that is still active. Initially, $u_1=v$. For each round $\tau$, if in round $\tau$, node $u_\tau$ broadcasts or it does not broadcast but it also does not receive a message, then let $u_{\tau+1}=u_{\tau}$. If $u{\tau}$ does not broadcast in round $c$ but it receives a message from a $G'$-neighbor $w$, then let $u_{\tau+1}=w$.
It follows from this recursive definition that $u_{4\log n}$ is active at the end of round $4\log n$ of the election part and hence, $u=u_{4\log n}$ joins the MIS. Furthermore, it is easy to see that $\left\|p(u_{\tau+1})-p(u_{\tau})\right\|_2 \leq \tau$ and hence, using the triangular inequality, we have $\left\|p(u)-p(v)\right\|_2 = \left\|p(u_{4\log n})-p(u_{1})\right\| \leq \sum_{\tau=1}^{4\log n-1} \left\|p(u_{\tau+1})-p(u_{\tau})\right\|_2 \leq 4c\log n$. This completes the proof.
\end{proof}



\begin{lemma}
The set $S$ of nodes that join the MIS is a maximal $G$-independent set {with high probability}.
\label{lem:fmmb:mis}
\end{lemma}
\begin{proof}
The proof requires us to establish two properties: (A) w.h.p., no two nodes $v, u\in S$ are $G$-neighbors, and (B) w.h.p., each node $v \in V\setminus S$ has a $G$-neighbor in $S$. Property (A) follows directly from \Cref{lem:independence}. We now prove property (B). Consider a node $v \in V$, suppose that $v$ does not join the MIS, and let $S'$ be the set of nodes within distance $O(c\log n)$ of $v$ that join the MIS. Using \Cref{lem:MISgrowth}, we know that for each phase $t$ in which $v$ starts as an active node, at least one new node joins $S'$. On the other hand, from \Cref{lem:independence}, we know that the set of nodes that join MIS and thus also $S'$ is w.h.p. a $G$-independent set. Hence, using \Cref{lem:Sphere-Packing}, we get that $|S'|=O(c^2\log^2 n)$. It follows that node $v$ cannot be active at the start of more than $O(c^2\log^2 n)$ phases, which means that there is a phase in which $v$ becomes permanently inactive. Recalling the description of the algorithm, we get that this means that in the announcement part of that phase, $v$ receives the message of a $G$-neighbor that has just joined the MIS. Hence, $v$ indeed has a $G$-neighbor in the MIS set $S$, which proves property $B$.
\end{proof}


}

\subsection{The Message Gathering Subroutine}\label{subsec:Gather-in-MIS}
We now describe a message gathering subroutine (called ``algorithm" in the rest of this subsection) that delivers each MMB message to a nearby MIS node in $O(c^2(k+\log n))$ rounds, w.h.p.
In more detail, each node $v$ maintains message-set $M_v\subseteq \mathcal{M}$ of messages that the node currently owns.
When this algorithm is first called, these sets describe the initial assignment of MMB message to nodes.
 Throughout the algorithm, the message-set of MIS nodes grow while the message set of non-MIS nodes shrink.
 The goal is to arrive at a configuration where $\cup_{v\in S} M_v = \mathcal{M}$: at which point, all messages in $\mathcal{M}$ are owned by MIS nodes. 
The algorithm is divided into $O(c^2 (k+\log n))$ {\em periods},
 where each period consists of three rounds.  At the start of each period, each MIS node decides to be active with probability $1/\Theta(c^2)$, and inactive otherwise. Then, in the first round of the period, each active MIS node broadcasts its ID, announcing that it is active. In the second round, each non-MIS node $v$ that received a message from one of its $G$-neighbors in the first round and has at least one message left in its message-set $M_v$ broadcasts one of the messages in $M_v$, along with its own ID. In the same round, if an MIS node $u$ receives a message $m$ from a $G$-neighbor, then node $u$ updates its message-set as $M_{u}=M_{u} \cup \{m\}$. In the third round of the period, each MIS node $u$ that received a message $m$ in the second round sends an \emph{acknowledgment} message, which contains message $m$ and its own $ID$. In this round, if a non-MIS node $v$ receives a message $m$ from a $G$-neighbor, then $v$ updates its message-set as $M_{v}=M_{v} \setminus \{m\}$.

\fullOnly{

\begin{lemma}
When the above algorithm is executed given a valid MIS $S$, the following holds at termination, w.h.p.: $\cup_{v\in S} M_v = \mathcal{M}$. That is, each message is owned by at least one MIS node.
\label{lem:fmmb:gather}
\end{lemma}
\begin{proof}
Consider a non-MIS node $v$ and suppose that at the start of the algorithm, node $v$ has message-set $M_v=T_0 \neq \emptyset$. We show that at the end of the algorithm, w.h.p., each message $m\in T_0$ is held by at least one MIS node $u$. Fix $u$ to be one (arbitrary) $G$-neighbor of $v$ that is in the MIS set. Let $A_{u} \subseteq \mathcal{M}$ be the set of messages for which $u$ has broadcast an acknowledgment and this acknowledgment is received by all $G$-neighbors of $u$. We prove that in each period in which $M_{v} \neq \emptyset$, with probability at least $1/\Theta(c^2)$, $|A_{u}|$ increases by one.

Let $S_u$ be the set of all MIS nodes that are within distance $2c$ of $u$. Using \Cref{lem:Sphere-Packing}, we know that $|S_u|=O(c^2)$. Therefore, for each period $t$, the probability that $u$ is the only MIS node in $S_u$ that is active in period $t$ is at least $1/\Theta(c^2) (1-1/\Theta(c^2))^{O(c^2)} = 1/\Theta(c^2)$. Suppose that $u$ is the only MIS node in $S_u$ that is active in period $t$. Furthermore, assume that $M_{v} \neq \emptyset$. Then, in the second round of period $t$, the only $G'$-neighbors of $u$ that are broadcasting are in fact $G$-neighbors of $u$. This is because, consider a node $w$ that is a $G'$-neighbor of $u$ but not a $G$-neighbor of $u$ and suppose that $w$ is broadcasting in the second round of period $t$. Then an MIS $G$-neighbor $w'\neq u$ of $w$ must be active in this period. It follows that $\left\|p(w')-p(u)\right\|_2\leq \left\|p(w')-p(w)\right\|_2+\left\|p(w)-p(u)\right\|_2 \leq 1+ c$. Thus, $w' \in S_u$ which is in contradiction with the assumption that $u$ is the only active node in $S_u$. Now, in period $t$, node $v$ broadcasts a message in $M_v$. Note that by the description of the algorithm $M_v \cap A_u = \emptyset$. Hence, we conclude that $u$ receives a message $m$ from one of its $G$-neighbors and this message is not in $A_u$. In the third round of this period, $u$ acknowledges this message $m$. We claim that this acknowledgment is received by all $G$-neighbors of $u$, which means that $|A_{u}|$ increases by one. The reason is that, if a $G$-neighbor $w$ of $u$ does not receive the acknowledgment, it means that a $G'$-neighbor $w'\neq u$ of $w$ was broadcasting in the third round. By the description of the algorithm, we get that $w'$ is an active MIS node, and furthermore, $\left\|p(w')-p(u)\right\|_2\leq \left\|p(w')-p(w)\right\|_2+\left\|p(w)-p(u)\right\|_2 \leq c+ 1$, which means that $w' \in S_u$, which is a contradiction to the assumption that $u$ is the only active MIS node in $S_u$ in period $t$. Hence, we have established that in each period in which $M_{v} \neq \emptyset$, with probability at least $1/\Theta(c^2)$, $|A_{u}|$ increases by one. Hence, in expectation, after $O(k c^2 )$ such periods, $|A_u|\geq k$. That is, the set $M_v$ is emptied which means that for each message $m$ that was originally in $M_v$, $v$ has received an acknowledgment and thus, the message $m$ is now held by at least one MIS nodes. A basic application of Chernoff bound then shows that after $O(c^2 (k+\log n)) = O(c^2 (k+\log n))$, w.h.p. we have $|A_u|\geq k$ and thus, each message $m$ intially held by $v$ is now held by at least one MIS nodes. Taking a union bound over all non-MIS nodes $v$ then completes the proof.
\end{proof}


}

\subsection{The Message Spreading Subroutine}\label{subsec:MIS-to-All}
We conclude by describing the subroutine (``algorithm" in the following subsection) used by FMMB to efficiently spread the messages
gathered at MIS nodes to the full network.
This algorithm spreads the messages
to all nodes in the network in $O((D+k)\log n)$ rounds, w.h.p.
In more detail, in the following, let $S$ be the set of MIS nodes when this algorithm is executed.
Assume $S$ is a valid MIS.
Let $E_S$ be the set of unordered pairs $(v, u) \in E$ such that the hop distance of $u$ and $v$ in graph $G$ is at most $3$. Consider the overlay graph $\mathcal{H}=(S, E_S)$. The algorithm works by spreading
messages over $\mathcal{H}$.
For this purpose, we explain a simple procedure, that uses $O(\log n)$ rounds, and that achieves the following: Suppose that each node $v\in S$ starts this procedure with at most one message $m_v$. Then, at the end of this procedure, w.h.p., we have that $m_v$ is delivered to all $\mathcal{H}$-neighbors of $v$. We will then establish the
final upper bound of $O((D+k)\log n)$ rounds by combining this procedure with a standard pipelining argument applied to messages in $\mathcal{H}$.

\smallskip
\noindent {\bf The Local Broadcast Procedure on the Overlay.}
The algorithm consists of $O(c^2\log n)$ periods, each consisting of three rounds. In each period, each node $v$ decides to be active with probability $1/\Theta(c^2)$ and remains inactive otherwise. If a node $v\in S$ is active, it broadcasts its message $m_v$ in the first round, if it has a message $m_v$. For all the three rounds of the period, if a node $u \in V$ receives a message from a $G$-neighbor in one round, it broadcasts this message in the next round. At the end of the three rounds of the period, each node $u\in S$ adds the messages that it has received to its message-set.

\fullOnly{ 

\begin{lemma}\label{lem:local-bcast}
At the end of the procedure, we have that for each node $v\in S$, if $v$ starts the procedure with message $m_v$, then $m_v$ is delivered to all $\mathcal{H}$-neighbors of $v$ with high probability.
\end{lemma}
\begin{proof}
Let $S_v$ be the set of nodes $u\in S$ such that $\left\|p(v)-p(u)\right\|_2 \leq 7c$. For now suppose that $v$ is the only node in $S_v$ that is active. We claim that in this case, in the $\tau^{th}$ round of the period---where $\tau\in \{1,2,3\}$, all nodes that their $G$-distance to $v$ is $\tau$ hops receive $m_v$. Hence, overall the three rounds, all $\mathcal{H}$-neighbors of $v$ receive $m_v$. The proof of this claim is as follows. First consider the case $\tau=1$. Then, if there is a $G$-neighbor $w$ of $v$ such that $w$ does not receive $m_v$ in the first round, it would mean that $w$ has a $G'$-neighbor $w'$ that is in $S$ and is active in this period. We have $\left\|p(w')-p(v)\right\|_2\leq \left\|p(w')-p(w)\right\|_2+\left\|p(w)-p(v)\right\|_2 \leq c+ 1$. Thus, $w$ is in $S_v$ which is in contradiction with the assumption that $v$ is the only node in $S_v$ that is active. Now we move to proving the claim for $\tau=2$ or $\tau=3$. Suppose that $\tau^*$ is the smallest $\tau \in \{2,3\}$ for which the claim breaks and there is a node $w$ that has $G$-distance of $\tau$ from $v$ but it does not receive $m_v$ in round $\tau$. We know that $w$ has a $G$-neighbor $w'$ that has $G$-distance $\tau^*-1$ from $v$ and $w'$ receives $m_v$ in round $\tau^*-1$. Hence, there must be a $G'$-neighbor $w''$ of $w$ that broadcasts a message $m'\neq m_v$ in round $\tau^*$. Given the description of the algorithm, it follows that there is an active node $u \in S$ which started message $m'$ in this period and $u$ is has $G$-distance at most $\tau^*$ from $w$. Thus, we get $\left\|p(v)-p(u)\right\|_2\leq \left\|p(v)-p(w'')\right\|_2+\left\|p(w'')-p(w)\right\|_2 + \left\|p(w)-p(v)\right\|_2 \leq \tau^*+ c+\tau^* \leq c+6\leq 7c$. This means that $w''$ is in $S_v$ which is in contradiction with the assumption that $v$ is the only node in $S_v$ that is active. This contradiction completes the proof of the claim, establishing that if $v$ is the only node in $S_v$ that is active, then $m_v$ is delivered to all $\mathcal{H}$-neighbors of $v$. Now note that using \Cref{lem:Sphere-Packing}, we get $|S_v|=O(c^2)$. Thus, in each period, the probability that $v$ is the only node in $S_v$ that is active is $1/\Theta(c^2) (1-1/\Theta(c^2))^{O(c^2)} = 1/\Theta(c^2)$. Hence, in $O(c^2\log n)$ periods, with high probability, there is at least one period in which $v$ is the only node in $S_v$ that is active. Therefore, with high probability, $m_v$ gets delivered to all $\mathcal{H}$-neighbors of $v$. Taking a union bound over all choices of node $v$ completes the proof.
\end{proof}

}

This local broadcast on the overlay provides essentially the same guarantee as given by $\fack$ on the full network topology, but with respect to the overlay graph $\mathcal{H}$. Having this {\em simulated broadcast},
the problem can be solved by combining BMMB with this simulated broadcast,
and then analyzing its performance with respect to $\mathcal{H}$.
That is, we divide the time into phases, each of length $O(\log n)$ rounds, where the constants are such that one run of the above procedure fits in one phase. Then, in each phase, each MIS node sends a message that it has not sent so far, to all of its $\mathcal{H}$-neighbors. It follows from \Cref{thm:bmmb:arbitrary} that after $O(D_{\mathcal{H}}+k)$ phases, all messages are broadcast over $\mathcal{H}$, i.e., to all MIS nodes. Here $D_{\mathcal{H}}$ is the hop diameter of the overlay graph $\mathcal{H}$, and we clearly have $D_{\mathcal{H}} \leq D_{G}=D$.
\shortOnly{Below is a  more detailed description of this part.}
\fullOnly{Below is a  more detailed description of this part, as well
as the final lemma statement for this subroutine.}

\smallskip
\noindent {\bf Broadcast on the Overlay Graph $\mathcal{H}$.} Here, we explain a more detailed version of the algorithm that broadcasts messages on the overlay graph $\mathcal{H}$, in $O((D+k)\log n)$ rounds. We divide the $O((D+k)\log n)$ rounds into $O(D+k)$ phases, each of length $O(\log n)$ rounds, where the constants are such that one run of the above procedure fits in one phase. In the algorithm, each node $v\in S$ has a message-set $M_v$ of messages that it has or it has received, and it also has a \emph{sent-set} $M'_v$ of messages that contains all the messages that $v$ has sent throughout this algorithm. Initially, for each node $v$, $M'_v=\emptyset$. In each phase, each node $v$ sets $m_v$ to be equal to one of the messages in $M_v\setminus M'_v$ and runs the procedure explained above. At the end of the phase, node $v$ adds $m_v$ to $M'_v$ and it also adds each message received during this phase to $M_v$. The following theorem shows that this algorithm broadcasts all messages to all MIS nodes.

\fullOnly{

\begin{lemma}\label{lem:fmmb:spread}
At the end of $D_{\mathcal{H}}+k$ phases, for each node $v \in S$, we have $M_v=\mathcal{M}$. Here $D_{\mathcal{H}}$ is the hop diameter of the overlay graph $\mathcal{H}$ and we clearly have $D_{\mathcal{H}} \leq D_{G}=D$
with high probability.
\end{lemma}

\begin{proof} 
Consider a message $m\in \mathcal{M}$ and let $S_m\subseteq S$ be the set of nodes $u\in S$ that hold $m$ at the start of the algorithm.
For each node $v\in S$, each $d\in [1, D_\mathcal{H}]$, and $\ell\in [1, k]$, set $t_{d_v, \ell}(v)=d_v+\ell$, where $d_v$ is the $\mathcal{H}$-distance of node $v$ to the set $S_m$, that is, the smallest $d$ such that there is a node $u\in S_m$ that is within $d$ $\mathcal{H}$-hops of $v$.

We claim that for each node $v$, after $t_{d, \ell}(v)$ phases, node $v$ has $m$ or $\ell$ other messages in its sent-set $M'_v$, w.h.p. It would then immediately follow that after $t_{D_\mathcal{H}, k}(v)$, node $v$ has $m$ in its sent-set $M'_v$ and hence also in $M_v$.

We prove the claim using an induction on $h=d_v+\ell$. The base case $h=0$ is straightforward as when $h=0$, we also have $d_v=0$ and in that case, the claim reduces to a trivial statement about the local queue of node $v$: Namely that if node $v$ has message $m$ in its local queue at the start of the algorithm, then after $\ell$ phases, $v$ has either message $m$ or $\ell$ other messages in its sent-set $M'_v$. For the inductive step, consider a node $v\in S$ such that $d_v+\ell=h$. If $d_v=0$, then the claim follows from the same trivial local-queue argument. Suppose that $d_v\geq 1$ and consider an $\mathcal{H}$-neighbor $u$ of $v$ such that $d_{u}=d_{v}-1$. By the induction hypothesis, we know that by the end of phase $h-1=d_v+\ell-1$, $v$ has either $m$ or at least $\ell-1$ other messages in $M'_v$, and $u$ has either $m$ or at least $\ell$ other messages in $M'_u$. For each of these four possibilities, we get that with high probability, by the end of phase $h=d_v+\ell$, node $v$ has either $m$ or at least $\ell$ other messages in $M'_v$. This is because of the following: if $v$ already has $m$ or at least $\ell$ other messages in $M'_v$ at the end of phase $h-1=$, then we are done. Otherwise, using \Cref{lem:local-bcast}, we get that w.h.p., by the end of phase $h-1$, node $v$ has received either $m$ or $\ell$ other messages from $u$ which shows that at the start of phase $h$, node $v$ has at least one message in $M_v\setminus M'_v$, either $m$ or a different message. Thus, at the end of phase $h$, either $m$ or $\ell$ other messages are in $M'_v$. This finishes the proof.
\end{proof}

}


\section{Conclusion}


In this paper, we applied the abstract MAC layer approach to a natural problem: disseminating an unknown
amount of information starting at unknown devices through an unknown network (what we call multi-message broadcast).
We proved that the presence of unreliable links has a significant but perhaps unexpected impact on the worst-case performance
of multi-message broadcast.
In particular, with no unreliability or unreliable links limited to nodes close in the reliable link graph,
basic flooding (what we called the BMMB algorithm) is efficient.
Once we shift to the similar constraint of unreliability limited to nodes
close in geographic distance, however, all solutions are inherently slow. This indicates
an interesting property of unreliability: the ability to unreliably connect nodes distant in the reliable link graph seems
to be what degrades worst-case performance of broadcast algorithms.
Finally, we demonstrated that if nodes have estimates of the model time bounds
and can abort messages in progress, even more efficient solutions to this problem are possible.
Most existing MAC layers do not offer an interface to abort messages. This result
motivates the implementation of this interface (which seems technically straightforward).

In terms of future work, 
there exist many other important problems for which a similar analysis can be performed,
such as leader election, consensus, and network structuring.
It would be interesting to investigate whether there are properties of link unreliability
that are universal to distributed computation in this setting, or if the properties of this type that matter differ
from problem to problem.
Another direction to study within this same general area is whether the strength of the scheduler
strongly impacts worst-case performance. In our lower bound, for example, the scheduler knows
the algorithm's random bits. This is a strong assumption and motivates the question of whether
this bound can be circumvented with a weaker adversary and a more clever algorithm.

\bibliographystyle{plain}
\bibliography{wireless}
\end{document}